\documentclass[11pt]{article}
\usepackage{booktabs}
\usepackage[utf8]{inputenc}

\usepackage{amsmath,amsfonts,amsthm,mathrsfs,xspace,graphicx}
\usepackage[backref,colorlinks,bookmarks=true, citecolor=blue, urlcolor=black]{hyperref}
\usepackage{endnotes}
\usepackage{color}
\usepackage{float}
\usepackage{xcolor}
\definecolor{since}{rgb}{0.5,0.5,0.5}
\definecolor{newred}{HTML}{ED2024}
\definecolor{newgreen}{HTML}{109A48}
\definecolor{newblue}{HTML}{535DAA}
\definecolor{neworange}{HTML}{F79420}
\usepackage{mdframed}
\usepackage{bbm}
\usepackage{suffix} \usepackage{times}
\usepackage{tabularx}
\usepackage{makecell}
\usepackage{amssymb,latexsym}
\usepackage[capitalize]{cleveref}
\usepackage[font=footnotesize,labelfont=bf]{caption}
\usepackage{enumitem}
\usepackage{tikz}
\usepackage{tikz-cd}
\usepackage{multirow}

\usepackage[section]{placeins}
\usepackage[affil-it]{authblk}

\makeatletter
\renewcommand*\env@matrix[1][*\c@MaxMatrixCols c]{%
  \hskip -\arraycolsep
  \let\@ifnextchar\new@ifnextchar
  \array{#1}}
\makeatother

\usepackage[margin=1in]{geometry}

\usepackage{endnotes}
\usepackage{amssymb,latexsym}
\usepackage{enumitem}
\setlist{itemsep=0mm}
\usepackage{tikz}
\usepackage{float}
\usepackage{setspace}
\usepackage{soul}

\usepackage{thmtools}
\usepackage{thm-restate}

\usepackage{mathtools, physics}
\usepackage{complexity}

\newclass{\QPCP}{QPCP}
\newclass{\QCPCP}{QCPCP}
\newclass{\QCMAcomp}{QCMA-complete}
\newclass{\sharpP}{\#P}

\usepackage[linesnumbered,ruled,vlined]{algorithm2e}
\usepackage{verbatim}

\DeclarePairedDelimiter\floor{\lfloor}{\rfloor}

\usepackage{amsthm}
\newtheorem{theorem}{Theorem}
\newtheorem*{theorem*}{Theorem}
\newtheorem{proposition}[theorem]{Proposition}
\newtheorem*{proposition*}{Proposition}
\newtheorem{fact}[theorem]{Fact}
\newtheorem*{fact*}{Fact}
\newtheorem{lemma}[theorem]{Lemma}
\newtheorem*{lemma*}{Lemma}

\newtheorem{corollary}[theorem]{Corollary}
\newtheorem{conjecture}{Conjecture}
\newtheorem*{conjecture*}{Conjecture}

\theoremstyle{definition}
\newtheorem{definition}[theorem]{Definition}
\newtheorem*{definition*}{Definition}

\theoremstyle{remark}

\newtheorem*{remark*}{Remark}
\newtheorem*{example}{Example}

\newcommand{\CC}{\ensuremath{\mathbb{C}}}

\newcommand{\FF}{\ensuremath{\mathbb{F}}}

\newcommand{\mcC}{\ensuremath{\mathcal{C}}}
\newcommand{\mcG}{\ensuremath{\mathcal{G}}}
\newcommand{\mcO}{\ensuremath{\mathcal{O}}}
\newcommand{\mcP}{\ensuremath{\mathcal{P}}}
\newcommand{\mcS}{\ensuremath{\mathcal{S}}}
\newcommand{\mcT}{\ensuremath{\mathcal{T}}}

\DeclareMathOperator\im{im}
\DeclareMathOperator\CSS{CSS}
\DeclareMathOperator\N{N}
\DeclareMathOperator\stab{Stab}

\DeclareMathOperator\evenodd{par}
\DeclareMathOperator\Span{Span}
\newcommand{\ham}{\mathcal{H}}

\newcommand{\destab}{\ensuremath{e^{-i\frac{\pi}{8}Y}}}
\newcommand{\destabn}{\ensuremath{(\destab)\n}}
\newcommand{\destabdagger}{\ensuremath{e^{i\frac{\pi}{8}Y}}}
\newcommand{\destabdaggern}{\ensuremath{(\destabdagger)\n}}

\newcommand{\Hn}{\ensuremath{\ham^{(n)}}}
\newcommand{\dHn}{\ensuremath{\Tilde{\ham}^{(n)}}}
\newcommand{\Hzero}{\ensuremath{\ham_0^{(n)}}}
\newcommand{\dHzero}{\ensuremath{\Tilde{\ham}_0^{(n)}}}
\newcommand{\ketzero}{\ensuremath{\ket{0}^{\otimes n}}}
\newcommand{\brazero}{\ensuremath{\bra{0}^{\otimes n}}}
\newcommand{\n}{\ensuremath{^{\otimes n}}}

\newcommand{\local}{\ensuremath{C}}

\newcommand{\NLXS}{\text{NL$X$S} }

\DeclareMathOperator\wt{wt}

\DeclareMathOperator\had{H}
\DeclareMathOperator\eye{\mathbb{I}}
\DeclareMathOperator\phase{P}
\DeclareMathOperator\T{T}
\DeclareMathOperator\CNOT{CNOT}

\usepackage{accents}

\newcommand\restr[2]{{
  \left.\kern-\nulldelimiterspace 
  #1 
  \right|_{#2} 
  }}

\newcommand{\nnote}[1]{}
\newcommand{\mnote}[1]{}
\newcommand{\jnote}[1]{}
\newcommand{\snote}[1]{}

\title{\vspace{-2.5em}Local Hamiltonians with no low-energy stabilizer states}

\author[1]{Nolan J. Coble \footnote{\href{mailto:ncoble@terpmail.umd.edu}{ncoble@terpmail.umd.edu}, \href{mailto:mcoudron@umd.edu}{mcoudron@umd.edu}, \href{mailto:nelson1@umd.edu}{nelson1@umd.edu}, \href{mailto:sajjad@umd.edu}{sajjad@umd.edu}}}
\author[1,2]{Matthew Coudron}
\author[1]{Jon Nelson}
\author[1]{Seyed Sajjad Nezhadi}

\affil[1]{Joint Center for Quantum Information and Computer Science (QuICS), University of Maryland}
\affil{Department of Computer Science, University of Maryland}
\affil[2]{National Institute of Standards and Technology}

\date{}

\begin{document}

\maketitle

\begin{abstract}
The recently-defined No Low-energy Sampleable States (NLSS) conjecture of Gharibian and Le Gall \cite{GL22} posits the existence of a family of local Hamiltonians where all states of low-enough constant energy do not have succinct representations allowing perfect sampling access. States that can be prepared using only Clifford gates (i.e. stabilizer states) are an example of sampleable states, so the NLSS conjecture implies the existence of local Hamiltonians whose low-energy space contains no stabilizer states. We describe families that exhibit this requisite property via a simple alteration to local Hamiltonians corresponding to CSS codes. Our method can also be applied to the recent NLTS Hamiltonians of Anshu, Breuckmann, and Nirkhe \cite{ABN22}, resulting in a family of local Hamiltonians whose low-energy space contains neither stabilizer states nor trivial states. We hope that our techniques will eventually be helpful for constructing Hamiltonians which simultaneously satisfy NLSS and NLTS.

\end{abstract}
\vspace{1em}
{
\tableofcontents 
}
\newpage

\section{Introduction}
Local Hamiltonians are ubiquitous in quantum physics and quantum computation. From the physical perspective, Hamiltonians describe the dynamics and energy spectra of closed quantum systems, with ``local'' Hamiltonians corresponding to models where only a small number of particles can directly interact with each other. From the computational perspective, local Hamiltonians naturally generalize well-studied constraint satisfaction problems through the ``local Hamiltonian problem'', which asks about the complexity of approximating the ground-state energy of local Hamiltonians.
\begin{definition*}[LH-$\delta(n)$]
    A $k$-local Hamiltonian, $\ham=\frac{1}{m}\sum_{i=1}^m\ham_i$, is a sum of $m=\poly(n)$ Hermitian matrices, $\ham_i\in \CC^{2^n\times 2^n}$, where each $\ham_i$ acts non-trivially on at most $k=\mcO(1)$ qubits\footnote{$\ham_i = h_i\otimes \eye_{2^{n-k}}$ where $h_i$ is a $2^k\times 2^k$ Hermitian matrix and $\eye_{2^{n-k}}$ is the $2^{n-k}\times 2^{n-k}$ identity matrix} and has bounded spectral norm, $\|\ham_i\|\leq 1$.
    
    Given a local Hamiltonian, $\ham$, and two real numbers $a<b$ with $b-a>\delta(n)$, the \textbf{local Hamiltonian problem with promise gap $\delta(n)$} is to decide if (1) there is a state with energy $\bra{\psi_0}\ham\ket{\psi_0}\leq a$ or (2) all states have energy $\bra{\psi}\ham\ket{\psi}\geq b$, given that one of these cases is true.\footnote{This is equivalent to deciding if $\ham$ has an eigenvalue less than $a$ or if all of the eigenvalues of $\ham$ are larger than $b$, which is the more typical formulation of the problem.} The value $\delta(n)$ is called the promise gap of the problem.
\end{definition*}

LH is a natural quantum analogue of the $\NP$-complete constraint satisfaction problem (CSP):\footnote{Technically LH is a generalization of the decision problem MAX-$k$-CSP.} the local terms serve as \textit{quantum} constraints on an $n$-qubit state, and the energy of a local term corresponds to how well the state satisfies that local constraint. The lowest energy state— or ground-state— of $\ham$ is the state that optimally satisfies all of the local constraints.

It is straightforward to show that CSP is $\NP$-complete for a promise gap $\delta(n) = 1/\poly(n)$, and the celebrated classical PCP Theorem \cite{AS92,ALM+98} shows that [surprisingly] CSP is still $\NP$-complete when $\delta(n) = \Omega(1)$, a constant. Since LH is the quantum generalization of a CSP we can similarly ask whether it is complete for the class QMA, the quantum version of NP.
Kitaev showed that LH is $\QMA$-complete for $\delta(n) = 1/\poly(n)$ when he originally defined the class of $\QMA$ problems \cite{KSV02}. Perhaps the most important open question in quantum complexity theory is whether or not a quantum version of the PCP theorem holds. The ``quantum PCP conjecture'' \cite{AN02,AAL+08} states that LH with a constant promise gap is $\QMA$-hard; the conjecture has thus far eluded proof.

As a possible step towards proving quantum PCP, Freedman and Hastings suggested the No Low-energy Trivial States (NLTS) conjecture which is implied by the quantum PCP conjecture (assuming $\NP\neq\QMA$). A local Hamiltonian has the NLTS property if there is a constant strictly larger than the ground-state energy which lower bounds the energy of any state preparable in constant-depth (“trivial states”). The NLTS conjecture posits the existence of an NLTS Hamiltonian. This seemingly simpler problem remained open for nearly a decade until Anshu and Breuckmann solved the combinatorial version \cite{AB22}, followed shortly after by a complete proof by Anshu, Breuckmann, and Nirkhe \cite{ABN22}. They explicitly constructed an NLTS Hamiltonian using recently developed asymptotically-good quantum LDPC codes \cite{LZ22}.

While the NLTS Theorem makes significant progress, there are still many other properties that a candidate Hamiltonian must satisfy in order to be $\QMA$-hard with a constant promise gap. For instance, Gharibian and Le Gall defined the No Low-energy Sampleable States (NLSS) conjecture \cite{GL22}. A state, $\ket\psi$ is ``sampleable'' if a classical computer can efficiently draw an $x\in\{0,1\}^n$ from the distribution defined by $p(x) = \abs{\bra{x}\ket{\psi}}^2$ and can calculate all of the amplitudes, $\bra{x}\ket{\psi}$.\footnote{The more proper terminology, as in \cite{GL22}, would be that $\ket\psi$ has a succinct representation allowing perfect sampling access. We will not be directly addressing the NLSS conjecture, so we will use the term ``sampleable'' for brevity.} A local Hamiltonian has the NLSS property if there is a constant which lower-bounds the energy of every sampleable state. The NLSS conjecture posits the existence of an NLSS Hamiltonian, and Gharibian and Le Gall showed that unless $\MA=\QMA$ the quantum PCP conjecture implies the NLSS conjecture.

In this paper we examine a simplified version of the NLSS conjecture, where instead of sampleable states we consider stabilizer states. A stabilizer state is the unique state stabilized by a commuting subgroup of the Pauli group with size $2^n$. Equivalently, stabilizer states are those states that can be prepared using only Clifford gates, i.e. Hadamard, Phase, and CNOT gates. We say that a local Hamiltonian has the No Low-energy Stabilizer States (NLCS)\footnote{The ``C'' in NLCS stands for Clifford, since states prepared by Clifford circuits and stabilizer states are equivalent.} property if there is a constant which lower-bounds the energy of any stabilizer state.\footnote{The existence of NLCS Hamiltonians has been suggested before as a direct consequence of the quantum PCP conjecture, for instance in \cite{AN22}. We discuss the relationship of NLCS and more to the quantum PCP conjecture in Section \ref{sec:qPCP-implications}.} The Gottesman-Knill Theorem \cite{Got98} shows that any stabilizer state can be efficiently sampled, so any NLSS Hamiltonian must also be an NLCS Hamiltonian. We show that a generic construction can be used to produce many NLCS Hamiltonians.

To prove the NLCS property for a particular local Hamiltonian one must show an explicit lower bound on the energy of all stabilizer states. Let $\ham=\frac{1}{m}\sum \ham_i$ be a local Hamiltonian and let $\ket\psi$ be an $n$-qubit state. The energy of any particular Hamiltonian term can be expressed as $\bra{\psi}\ham_i\ket{\psi}=\Tr\big[\psi_{A_i} h_i\big]$, where $A_i$ is the set of qubits where $\ham_i$ acts non-trivially, $\psi_{A_i}$ is the reduced state of $\ket\psi$ on $A_i$, and $h_i$ is the non-trivial part of $\ham_i$. Suppose for simplicity that $\abs{A_i}=k$ for all $i$. One particularly strong way to lower-bound the energy of $\ket\psi$ would be to ``locally'' bound each energy term. That is, prove that each $\Tr\big[\psi_{A_i} h_i\big]$ is lower-bounded by a constant. In general this is not an easy task. However, stabilizer states have a rather convenient property: we show in Claim \ref{clm:Clifford-decomp} that if $\ket{\psi}$ is a stabilizer state, then every $\psi_{A_i}$ is a convex combination of stabilizer states on $k$ qubits. Thus, to lower-bound $\Tr\big[\psi_{A_i} h_i\big]$ for every $n$-qubit stabilizer state, $\ket\psi$, it is sufficient to lower-bound the quantity $\bra{\zeta}h_i\ket{\zeta}$ for every \textit{$k$-qubit} stabilizer state $\ket\zeta$.

This observation leads to a rather simple NLCS Hamiltonian. First, consider the Hamiltonian ${\ham_0} = \frac{1}{n}\sum\ketbra{1}_i$ where $\ketbra{1}_i$ is the projector to $\ket{1}$ on the $i$-th qubit and identity elsewhere. All of the local terms are the single-qubit projector $\ketbra{1}$. Clearly, we cannot lower-bound the energy of stabilizer states since $\ket{0}$ has energy 0. We can fix this, however, by instead considering a ``conjugated'' version of $\ham_0$: 
\begin{equation*}
    \Tilde{\ham}_0 \equiv \frac{1}{n} \sum_{i=1}^n  \Big(\destabdagger\ketbra{1}\destab\Big)\mid_i,
\end{equation*}
which can alternatively be expressed as $\Tilde{\ham}_0 = \destabdaggern \ham_0 \destabn$. Each local term is the single-qubit projector $\destabdagger\ketbra{1}\destab$, and it is straightforward to calculate that every single-qubit stabilizer state has high energy under this local term. We give a self-contained proof that $\Tilde{\ham}_0$ is NLCS in Appendix \ref{app:NLCS}.

The quantum PCP conjecture not only implies the existence of NLTS/NLCS/NLSS Hamiltonians, but also the existence of \textit{simultaneous} NLTS/NLCS/NLSS Hamiltonians. The process of conjugating a local Hamiltonian by a low-depth circuit conveniently preserves the NLTS property. That is, if $\ham$ is NLTS and $C$ is a constant-depth circuit, then $C^\dagger \ham C$ is also NLTS (see Lemma \ref{lem:NLTS-phase}). 

We note that since $\ketbra{1}=\frac{1}{2}(\eye-Z)$ the Hamiltonian $\ham_0$ is an example of a \textit{CSS Hamiltonian}, i.e. the local Hamiltonian terms are of the form $\frac{1}{2}(\eye-P_i)$ where the $P_i$'s are commuting $X$ and $Z$ type Pauli operators.
As the Hamiltonian $\Tilde{\ham}_0$ is simply $\ham_0$ conjugated by a depth-1 circuit $(\destab)\n$ it may be natural to ask whether the same procedure can be done to the NLTS Hamiltonians from \cite{ABN22} as they are also CSS Hamiltonians. The main result of our paper is the following:

\begin{theorem}[Informal version of Theorem \ref{thm:NLTS-NLCS}]
    Let $\ham_{NLTS}$ be the NLTS local Hamiltonian from \cite{ABN22}. The local Hamiltonian given by $\Tilde{\ham}_{NLTS}\equiv (\destabdagger)\n\ham_{NLTS}(\destab)\n$ satisfies both NLTS and NLCS.
\end{theorem}

We prove Theorem \ref{thm:NLTS-NLCS} by exhibiting local lower bounds on the individual Hamiltonian terms. In particular, we show that if $h=\frac{1}{2}(\eye-P^{\otimes k})$ is a $k$-local term where $P\in\{X,Z\}$, then $$\bra\zeta (\destabdagger)^{\otimes k} h (\destab)^{\otimes k}\ket\zeta \geq\sin^2(\pi/8)$$ for \textit{every} $k$-qubit stabilizer state $\ket\zeta$, as long as $k$ is odd. Combining this lower bound with the fact that the reduced state of a stabilizer state is a convex combination of stabilizer states, we have that conjugating a CSS Hamiltonian by $\destabn$ results in an NLCS Hamiltonian, at least in the case that many of the Hamiltonian terms act on an odd number of qubits.

The condition of odd weight is unfortunately a necessary condition of our local techniques: if $k$ is even then there is always a $k$-qubit stabilizer state with $\bra{\zeta_0} (\destabdagger)^{\otimes k} h (\destab)^{\otimes k}\ket{\zeta_0}=0$. Nonetheless, we show in Section \ref{app:odd-weight} that there is an explicit NLTS Hamiltonian from \cite{ABN22} where every local term acts on an odd number of qubits. Since conjugating by a constant-depth circuit preserves NLTS, we ultimately have that $\Tilde{\ham}_{NLTS}$ satisfies both NLTS and NLCS.


\subsection{Implications of the quantum PCP conjecture}\label{sec:qPCP-implications}

We turn now to the question of what Hamiltonians are guaranteed to exist by the quantum PCP conjecture. The quantum PCP conjecture has two main formulations; we focus here on the gap amplification version. See \cite{AAV13} for a great survey on the conjecture.

\begin{conjecture*}[Conjecture 1.3 of \cite{AAV13}]
    Let $\epsilon>0$ be a constant. \emph{LH-}$\epsilon$ is $\QMA$\emph{-hard} under quantum polynomial-time reductions.
\end{conjecture*}
In other words, the conjecture says there is a worst-case local Hamiltonian whose ground state energy is $\QMA$-hard to approximate within a constant. Approximating ground-state energies and finding ground states of local-Hamiltonians are of central importance to condensed matter theory and quantum simulation algorithms. If true, the quantum PCP conjecture says that there are some Hamiltonians whose ground-state energies we could never hope to approximate, let alone find their ground states.\footnote{Unless, of course, one believes $\QMA\subseteq\P$ or some other weakening of $\QMA$.}

The key insight of \cite{FH13} when they defined the NLTS conjecture was that some states have properties which allow their ground state energies to be calculated in a smaller complexity class than $\QMA$.  For a constant, $k$, we say that an $n$-qubit state, $\rho$, is \textbf{$k$-locally-approximable} if it has a polynomial-sized classical description from which every $k$-local reduced state, $\rho_A\equiv\Tr_{-A}[\rho]$ where $\abs{A}\leq k$, can be approximated to inverse-polynomial precision in polynomial-time.
Consider the following simple result:

\begin{fact}\label{fact:classical-local-access}
    Suppose $\ham =\frac{1}{m}\sum_{i=1}^m \ham_i$ is a $k$-local Hamiltonian and $\rho$ is a $k$-locally approximable state. The energy of $\rho$ under $\ham$ can be approximated to inverse-polynomial precision in $\NP$.
\end{fact}
\begin{proof}
     Each $\ham_i$ acts non-trivially on at most $k$ qubits, $A_i\subset [n]$, so the energy of $\rho$ for $\ham_i$ is $\Tr[\rho\ham_i]=\Tr\big[\rho_{A_i} h_i\big]$, where $h_i \equiv\Tr_{-A_i}[\ham_i]$ is the non-trivial part of $\ham_i$. Since $h_i\in\CC^{2^k\times 2^k}$ and by assumption we can efficiently compute $\rho_{A_i}$ to inverse-polynomial precision from the classical description of $\rho$, each $\Tr[\rho\ham_i]$ can be brute-force approximated in polynomial-time.
\end{proof}

Trivial states are locally approximable.
If $\ket\psi$ is a trivial state then there is a constant-depth circuit such that $\ket\psi= C\ket{0}\n$. For a set of $k$ qubits, $A$, the only gates that contribute to $\psi_A$ are those in the reverse-lightcone\footnote{See Figure \ref{fig:rotated-Hamiltonian}(a).} of $A$. As the reverse-lightcone has size at most $k2^d$, a constant, only a constant number of gates from $C$ are needed to brute-force approximate $\psi_A$. Thus, we can approximate local reduced states of $\ketbra\psi$ from the classical description of $C$.

The assumption of being able to compute local reduced states also holds for stabilizer states. 
Suppose $\ket\psi$ is an $n$-qubit stabilizer state. Since $\ket\psi$ is a stabilizer state there are $n$ independent and commuting Pauli operators $\{P_1,\dots, P_n\}$ that stabilize $\ket\psi$. The list of these Pauli operators will serve as the classical description of $\ketbra\psi$ from which local reduced states can be computed. The reduced state $\psi_A$ can be written as
\begin{equation}\label{eq:stabilizer-reduced-state}
    \psi_A = \frac{1}{2^k}\sum_{P\in G_A} P,
\end{equation}
where $G_A$ is the subgroup of the stabilizers of $\ket\psi$ which act non-trivially only on qubits in $A$. There are $4^{k}$ such Pauli group elements (ignoring phases) which we denote by $\mcP_A$. For $P\in\mcP_A$, one of $\pm P$ is in the stabilizer group of $\ket\psi$ if and only if $P$ commutes with every stabilizer generator. So, we can determine the elements of $G_A$ by brute-force checking which elements of $\mcP_A$ commute with every generator.\footnote{It remains to determine whether $+P$ or $-P$ is in the stabilizer group. Although slightly more complicated, this can be done in polynomial-time \textit{independent} of the weight of $P$.} This computation can be done in polynomial-time since there are only a constant number of Pauli operators to check, so using Equation \eqref{eq:stabilizer-reduced-state} we can compute $\psi_A$ efficiently.

Thus, in addition to being an implication of NLSS, NLCS Hamiltonians are also implied by the quantum PCP conjecture assuming $\NP\neq\QMA$: if every local Hamiltonian has a low-energy stabilizer state then the ground state energy could be computed in $\NP$ via Fact \ref{fact:classical-local-access}.

\subsection{Acknowledgements}
This paper is a contribution of NIST, an agency of the US government, and is not subject to US copyright. We thank Alexander Barg and Chinmay Nirkhe for helpful discussions.


\section{Preliminaries}\label{sec:prelim}
For a natural number, $n$, we denote $[n]\equiv\{1,\dots,n\}$. For a subset, $A\subseteq [n]$, we denote the set complement by $-A\equiv [n]\setminus A$ and the partial trace over the qubits in $A$ by $\Tr_A$. In particular, $\Tr_{-A}[\ketbra{\psi}]$ denotes the local density matrix of $\ket\psi$ on the qubits in $A$.
\subsection{States}
Let $C=\{C_n\}$ be a countable family of quantum circuits consisting of one and two-qubits gates where each $C_n$ acts on $n$ qubits. If the depth of $C_n$ is upper bounded by a function $d(n)$ for all $n$, then we say $C$ is a \textbf{depth-$d(n)$} family of quantum circuits. If $d(n)=\mcO(1)$ then we say $C$ is a depth-$\mcO(1)$ (or constant-depth) family of quantum circuits. Similarly, if $d(n)=\poly(n)$ then we say $C$ is a depth-$\poly(n)$ (or polynomial-sized) family of quantum circuits.

The single-qubit \textbf{Pauli group} is the set $\mcP_1\equiv\{i^\ell P\mid P\in\{\eye,X,Y,Z\}, \ell\in\{0,1,2,3\}\}$, and the $n$-qubit Pauli group is its $n$-fold tensor-power, $\mcP_n=\bigotimes_{i\in[n]}\mcP_1$.
For an element $S=P_1\otimes\dots\otimes P_n\in\mcP_n$, the \textbf{weight} of $S$ is defined to be the number of qubits where $P_i$ is not identity, i.e. $\wt(S)=\abs{\{P_i\mid P_i\neq i^\ell\eye\}}$. 
We denote the set of these qubits where $S$ acts non-trivially by $\N(S)\subseteq[n]$.

The $n$-qubit \textbf{Clifford group}, $\mcC_n$, is the set of unitary operators which stabilize the Pauli group.
It is well-known that $\mcC_n$ is generated by the set $\{\had,\phase,\CNOT\}$, where $\had$ is the single-qubit Hadamard gate, $\phase$ is the single-qubit phase gate, and $\CNOT$ is the two-qubit controlled-NOT gate. A \textbf{Clifford circuit} is defined to be any element of the Clifford group.

Let $\psi$ be a [possibly mixed] state on $n$ qubits and let $N\geq n$. If there is a quantum circuit, $C$, acting on $N$ qubits such that $\psi=\Tr_{N}[C\ketbra{0^N}C^\dagger]$ then we say that $C$ \textbf{prepares} $\psi$. $\psi$ is said to be: a \textbf{trivial state} if there is a constant-depth quantum circuit preparing it, an [efficiently] \textbf{preparable state} if there is a polynomial-sized circuit preparing it, a \textbf{Clifford state} if there is a polynomial-sized Clifford circuit preparing it, and an \textbf{almost Clifford state} if there is a polynomial-sized quantum circuit containing Clifford + $\mcO(\log(n))$ $\T$-gates preparing it. A pure state, $\ket\psi$ is said to be a \textbf{sampleable state} if (1) there is a classical algorithm exactly computing $\bra{x}\ket{\psi}$ for every $x\in\{0,1\}^n$ and (2) there is a classical algorithm that exactly samples $x\in\{0,1\}^n$ from the distribution $p(x) = \abs{\bra{x}\ket{\psi}}^2$.

A \textbf{stabilizer group} is an abelian subgroup, $G$, of $\mcP_n$ not containing $-\eye$. As a finite group, we can always find a list of mutually independent and commuting generators, $\mcS=\{S_1,\dots,S_k\}$, of $G$.
We will refer simply to the subgroup $\langle\mcS\rangle=G$ when this generating set is clear. Note that given a stabilizer group, there is a well-defined \textbf{stabilizer code} \cite{Got96,CSS97,CSS98}, $\mcC_\mcS$, which is the common +1 eigenspace of the operators in $\langle\mcS\rangle$.

If a given stabilizer group has a generating set, $\mcS$, consisting of tensor products of only Pauli $X$ and $\eye$ or only Pauli $Z$ and $\eye$, then we say $\mcC_\mcS$ is a \textbf{CSS code} and that $\mcS$ generates a CSS code.

The \textbf{stabilizer group} of a pure state, $\ket{\psi}$, is the subgroup of the Pauli group defined by $\stab(\ket\psi)\equiv \left\{ P\in\mcP_n\mid P\ket{\psi} = \ket{\psi}\right\}$.
We say that a $P\in\stab(\ket{\psi})$ \textbf{stabilizes} $\ket{\psi}$. Note that $\stab(\ket{\psi})$ is an abelian subgroup of the Pauli group not containing $-\eye$, and so it is a valid stabilizer group as before.

A pure state, $\ket{\psi}$, is said to be a \textbf{stabilizer state} if $\abs{\stab(\ket\psi)}=2^n$, or equivalently, if there are $n$ independent Pauli operators that stabilize $\ket\psi$. We note that $\ketbra{\psi}=\frac{1}{2^n}\sum_{g\in G}g$ where $G=\stab(\ket\psi)$.

A mixed state, $\psi$, is said to be a stabilizer state if $\psi$ is a convex combination of pure stabilizer states, i.e. $\psi = \sum_j p_j\ketbra{\varphi_j}$ where each $\ket{\varphi_j}$ is a pure stabilizer state on $n$ qubits, $\sum_j p_j = 1$, and $p_j\geq 0$.

All of the states defined here are related to one another via the following:
\begin{equation}\label{eq:state-implications}
    \begin{tikzcd}
    \text{trivial}   & \text{Clifford/Stabilizer} &   \\
     &  \text{almost Clifford}\arrow[u, Leftarrow, "\text{some $\T$ gates}"']&  \\
     \text{preparable} \arrow[uu, Leftarrow, "\text{increase depth}"']\arrow[ur, Leftarrow, "\text{arbitrary $\T$ gates}"'] &  &\text{sampleable}\arrow[ul, Leftarrow, "\text{\cite{BG16}}"']   
    \end{tikzcd}
\end{equation}


By definition of the Clifford group, stabilizer states and Clifford states are equivalent for pure states. We will interchangeably use the terms ``stabilizer state'' and ``Clifford state'' even for mixed states, which is motivated by the following result:
\begin{restatable}{claim}{Cliffdecomp}\label{clm:Clifford-decomp}
If $\psi$ is a Clifford state, then it is also a stabilizer state.
\end{restatable}
A proof can be found in Appendix \ref{app:Clifford-decomp}. Claim \ref{clm:Clifford-decomp} says that the reduced state of a pure stabilizer state is a convex combination of pure stabilizer states on the subsystem. This is essential in our energy lower bound arguments: To lower-bound the energy of all $n$-qubit stabilizer states for a $k$-local term of the Hamiltonian, $\ham_i$, it is sufficient to lower-bound the energy of all $k$-qubit stabilizer states for the non-trivial part of $\ham_i$.

\subsection{Hamiltonians}\label{sec:props-Hamiltonians}
A \textbf{$k$-local Hamiltonian}, $\Hn$, is a Hermitian operator on the space of $n$ qubits, $({\CC^2})\n$, which can be written as a sum $\Hn =\frac{1}{m}\sum_{i=1}^m \ham_i$, where each $\ham_i$ is a Hermitian matrix acting non-trivially on only $k$ qubits and with spectral norm $\norm{\ham_i}\leq 1$. A \textbf{family of $k$-local Hamiltonians}, $\{\Hn\}$, is a countable set of $k$-local Hamiltonians indexed by system size, $n$, where $k=\mcO(1)$ and $m=\poly(n)$. We will often use the term ``local Hamiltonian'' to mean a family of $k$-local Hamiltonians.

The \textbf{ground-state energy} of $\ham$ is $E_0\equiv \min_\rho \Tr[\rho \ham]$, where the minimization is taken over all $n$-qubit mixed states. $\ham$ is said to be \textbf{frustration-free} if $E_0=0$. A state, $\psi$, is said to be a \textbf{ground state} of $\ham$ if $\Tr[\psi \ham]=E_0$. A state, $\psi$, is said to be an \textbf{$\epsilon$-low-energy state} of $\ham$ if $\Tr[\psi \ham] < E_0 + \epsilon$. If $\psi=\ketbra{\psi}$ is a pure state, this condition simplifies to $\bra{\psi}\ham\ket{\psi}< \lambda_{\min}(\ham) + \epsilon$, where $\lambda_{\min}(\ham)$ is the smallest eigenvalue of $\ham$. For frustration-free Hamiltonians this is equivalent to $\bra{\psi}\ham\ket{\psi}<\epsilon$. All of the Hamiltonians we consider will be frustration-free.

For $S\in\mcP_n$, we denote the orthogonal projector to the $+1$ eigenspace of $S$ by $\Pi_S$, i.e. $\Pi_S\equiv\frac{\eye-S}{2}$. Since $\Pi_S$ acts non-trivially on only $\wt(S)$ qubits, we can write
$\Pi_S = {\Pi_S}\mid_{\N(S)}\otimes \eye_{[n]\setminus \N(S)}$.

Given a stabilizer group, $\langle \mcS\rangle$, with generating set $\mcS$, the \textbf{stabilizer Hamiltonian} associated to $\mcS$ is $\ham_\mcS \equiv \frac{1}{\abs{\mcS}}\sum_{S\in\mcS} \Pi_S$.
If each qubit is acted on non-trivially by at most $\wt(\mcS)$ elements of $\mcS$, then $\ham_\mcS$ is a $\wt(\mcS)$-local Hamiltonian. If $\mcC$ is the Stabilizer code associated with $\mcS$, then every $\ket{\psi}\in\mcC$ is a zero-energy state of $\ham_\mcS$. In particular, $\ham_\mcS$ is frustration-free with ground-state space $\mcC$. If $\mcS$ generates a CSS code then we say $\ham_\mcS$ is a \textbf{CSS Hamiltonian}.

If $\{\langle\mcS_n\rangle\mid \langle\mcS_n\rangle\leq \mcP_n\}$ is a countable family of stabilizer groups then the \textbf{family of stabilizer (or CSS) Hamiltonians} associated with $\{\mcS_n\}$ is $\{\ham_{\mcS_n}\}$. This will be a family of local Hamiltonians when: (1) each qubit is acted on non-trivially by at most $\wt(\mcS_n)$ elements of $\mcS_n$, (2) $\wt(\mcS_n)=\mcO(1)$, and (3) $|\mcS_n|=\Theta(n)$. Such families, $\{\langle\mcS_n\rangle\}$, of stabilizer groups correspond to quantum LDPC code families.

For each of the states in the previous section we can consider an analogue of NLTS.

\begin{definition*}
A family of $k$-local Hamiltonians, $\{\Hn\}$, is said to have the \textbf{$\epsilon$-\NLXS} property if for all sufficiently large $n$, $\Hn$ has no $\epsilon$-low-energy states of type $X$. The family, $\{\Hn\}$, is said to have the \textbf{\NLXS} property if it is $\epsilon$-\NLXS for some constant $\epsilon$.
\end{definition*}
The following implications between the \NLXS theorems/conjectures and quantum PCP conjecture hold. A complexity inequality next to an arrow denotes an implication that holds if the separation is true, e.g. if the quantum PCP conjecture is true and $\MA\neq\QMA$, then NLSS is true.
\begin{equation}\label{diagram:qpcp-implications}
    \begin{tikzcd}
      & \text{qPCP conjecture}
      \arrow[dl, Rightarrow,"\QCMA\neq\QMA"',red] 
      \arrow[dr, Rightarrow, "\MA\neq\QMA\cite{GL22}",red!40!yellow] 
      \arrow[dddl, Rightarrow,blue] & \\
      \text{NLPS}\arrow[dd, Rightarrow]\arrow[dr, Rightarrow] & & \text{NLSS}\arrow[dl, Rightarrow]  \\
      &    \text{NLACS}\arrow[uu, Leftarrow, "\NP\neq\QMA", blue]\arrow[d, Rightarrow]&\\
      \text{NLTS} &   \text{NLCS} &
    \end{tikzcd}
\end{equation}
The relationships between each of the \NLXS results are implicitly given by Diagram \ref{eq:state-implications}. Trivial states, stabilizer states, and almost Clifford states are all examples of locally-approximable states, so they following from the quantum PCP conjecture via Fact \ref{fact:classical-local-access}. The implication of NLSS was given by Gharibian and Le Gall when they originally defined NLSS \cite{GL22}. The implication of NLPS is well-known: if every local Hamiltonian has a low-energy preparable state, $C\ket{0}\n$, then given the classical description of $C$ a quantum prover could simply prepare the state and measure its energy. This would put LH-$\epsilon\in\QCMA$, implying $\QMA=\QCMA$ if the quantum PCP conjecture is true.

For a family of $k$-local Hamiltonians, $\{\Hn\}$, and a family, $C=\{C_n\}$, of depth-$\mcO(1)$ quantum circuits, we define the \textbf{$C$-rotated version} of $\{\Hn\}$ as $\{\Hn\}^C \equiv \{{C_n}^{\dagger} \Hn {C_n}\}$.
This is still a family of local Hamiltonians, albeit with a possibly different $k$ than the original Hamiltonian. This is because the only qubits that interact non-trivially with a single Hamiltonian term, $C^\dagger \ham_i C$, are those qubits in the reverse-lightcone of the qubits acted on by $\ham_i$. The number of qubits in the reverse-lightcone of a single qubit grows exponentially in the depth of a circuit, which is still constant since $C$ is constant-depth. See Figure \ref{fig:rotated-Hamiltonian} for an example of this. When $C=V\n$ is the tensor-product of a single-qubit gate, $V$, we will use the term ``$V$-rotated'' as opposed to ``$V\n$-rotated''.

\begin{figure}[t!]
    \centering
    \includegraphics[width=6.5in]{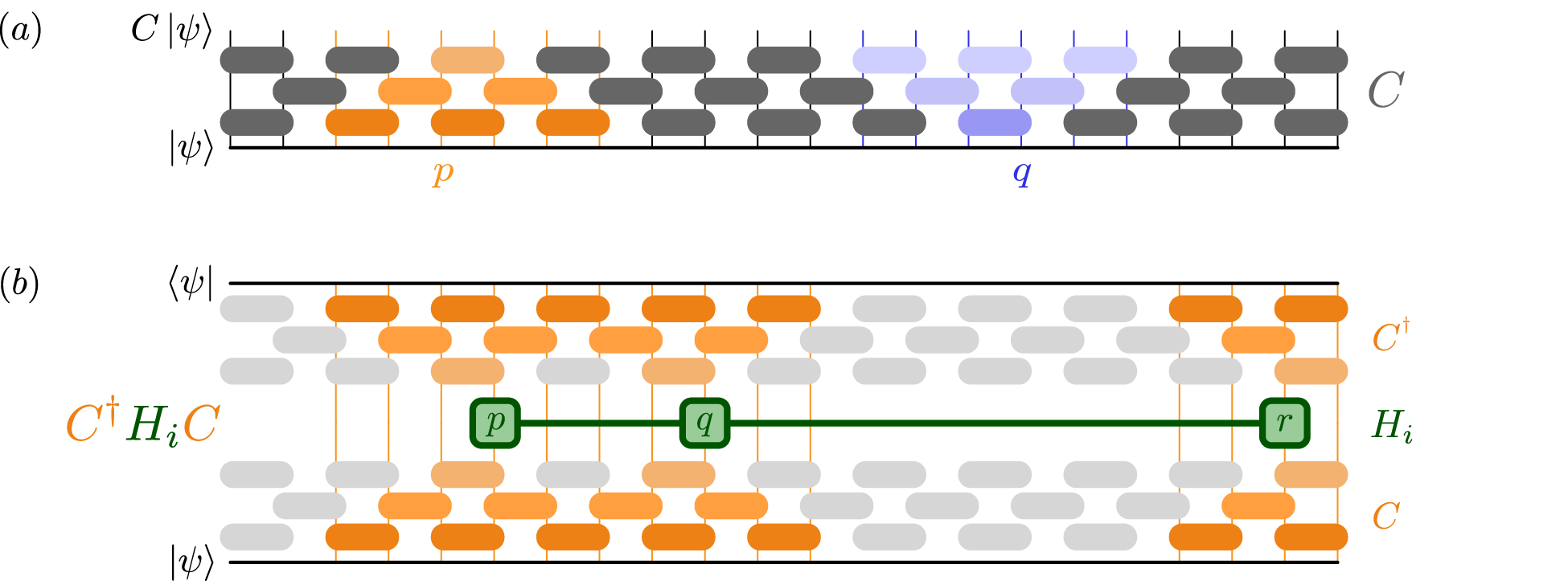}
    \caption{(a) Consider a constant-depth circuit, $C$. The [blue] highlighted gates on the right of the figure represent the ``lightcone'' of qubit $q$, i.e. the set of gates that can be traced back to $q$. The [orange] highlighted gates on the left of the figure represent the gates in the ``reverse-lightcone'' of qubit $p$, i.e. the gates that will ultimately affect $p$. \newline
    (b) Consider a single $k$-local Hamiltonian term, $\ham_i$, that acts only on qubits $p,q,$ and $r$. When conjugating $\ham_i$ with $C$, any gate not in the reverse-lightcone of one of $p,q,$ or $r$ will cancel out with its inverse. The number of qubits in the reverse-lightcone of any one qubit is $\leq 2^d$ where $d$ is the depth of $C$, so $C^\dagger \ham_i C$ will be at most $k2^d$-local. 
    }
    \label{fig:rotated-Hamiltonian}
\end{figure}
The utility of considering a $C$-rotated Hamiltonian is that in addition to preserving locality, the NLTS property is also preserved.

\begin{lemma}\label{lem:NLTS-phase}
    If $\{\Hn\}$ is a family of $\epsilon_0$-NLTS local Hamiltonians and $C=\{C_n\}$ is a family of constant-depth circuits, then $\{\Hn\}^C$ is also $\epsilon_0$-NLTS.
\end{lemma}
\begin{proof}
Suppose that $\{\Hn\}^C$ is not NLTS. By definition, for every $\epsilon >0$ there is an $n$ and constant-depth circuit $U_{\epsilon,n}$ such that $U_{\epsilon,n} \ketzero$ is an $\epsilon$-low-energy state of $C_n^\dagger\Hn C_n$, i.e.
\begin{equation*}
    \brazero U_{\epsilon,n}^\dagger C_n^\dagger\Hn C_n U_{\epsilon,n} \ketzero < \lambda_{\min}(C_n^\dagger\Hn C_n) + \epsilon.
\end{equation*}
Since $C_n$ is a unitary operator the minimum eigenvalues of $\Hn$  and $C_n^\dagger \Hn C_n$ are equal. Defining $\ket{\psi_{\epsilon_0,n}}\equiv C_n U_{\epsilon_0,n}\ketzero$ we have
\begin{equation*}
    \bra{\psi_{\epsilon_0,n}}\Hn\ket{\psi_{\epsilon_0,n}} < \lambda_{\min}(\Hn) + \epsilon_0,
\end{equation*}
i.e. $\ket{\psi_{\epsilon_0,n}}$ is an $\epsilon_0$-low-energy state of $\Hn$. Since $C_n U_{\epsilon_0,n}$ is a constant-depth circuit this implies that $\Hn$ has a low-energy trivial state, contradicting the assumption of $\epsilon_0$-NLTS.
\end{proof}


\section{NLCS from CSS codes}\label{sec:NLCS-CSS}
We will show that rotating by the tensor product of a single-qubit gate is sufficient to turn most CSS Hamiltonians into NLCS Hamiltonians, including the quantum Tanner codes used in \cite{ABN22}. In particular, we consider the single-qubit gate $D \equiv {\destab}$ and rotate a CSS Hamiltonian by $D\n$. For a local Hamiltonian, $\Hn$, we will denote its \textbf{$D$-rotated version} by $\dHn \equiv D^{\dagger\otimes n} \Hn D\n$. We denote the $D$-\textbf{rotated projector} associated with a Pauli element $S\in\mcP_n$ by $\Tilde{\Pi}_S\equiv D\n\Pi_S D^{\dagger\otimes n}$. By definition, we have $\Tilde{\Pi}_S = {\Tilde{\Pi}_S}\mid_{\N(S)}\otimes \eye_{[n]\setminus \N(S)}$, where $\Tilde{\Pi}_S\mid_{\N(S)} = D^{\otimes\wt(S)} \Pi_S\mid_{\N(S)} D^{\dagger\otimes\wt(S)}$. Note that we have not explicitly included $D$ in the above notations since $D$ will refer exclusively to ${\destab}$, throughout.

We have the following result:
\begin{restatable}{theorem}{oddWeightCSS}\label{thm:CSS-NLCS}
Let $\{\ham_{\mcS_n}\}$ be a family of CSS Hamiltonians associated with a family of quantum (CSS) LDPC codes, $\{\langle\mcS_n\rangle\}$. Suppose for every $n$ a constant fraction, $\alpha>0$, of the generators $S\in\mcS_n$ have odd weight. Then $\{\Tilde{\ham}_{\mcS_n}\}$ is a family of NLCS Hamiltonians.
\end{restatable}

We prove this by giving local lower bounds on the energies of $D$-rotated projectors associated with CSS generators. As a technical requirement, these lower bounds only hold when the weight of a generator is odd.

Recall that, up to a permutation of the qubits, the generators of a CSS code can be written as either $\bar X\otimes \eye$ or $\bar Z\otimes \eye$, where $\bar X\equiv X^{\otimes k}$ and $\bar Z\equiv Z^{\otimes k}$. First consider what happens to the projectors $\Pi_{\bar X}$ and $\Pi_{\bar Z}$ when rotating by $D$:

\begin{restatable}{claim}{destabXZ}\label{lem:destabilize-XZ}
    \begin{equation*}
        \Tilde{\Pi}_{\bar X} = \frac{\eye-\had^{\otimes k}}{2}, \hspace{5em} \Tilde{\Pi}_{\bar Z} = \frac{\eye-(-X\had X)^{\otimes k}}{2}.
    \end{equation*}
\end{restatable}

These identities are derived in Appendix \ref{app:lb-Xk-Zk}. The local lower bounds will be a result of the following:
\begin{lemma}\label{lem:H-up}
If $k$ is odd, then for every $k$-qubit stabilizer state, $\ket{\eta}$, we have $\abs{\bra{\eta}\had^{\otimes k}\ket{\eta}}\leq \frac{1}{\sqrt{2}}$. On the other hand, if $k$ is even then there exists a $k$-qubit stabilizer state, $\ket{\eta_0}$, with $\bra{\eta_0}\had^{\otimes k}\ket{\eta_0}=1$.
\end{lemma}
The proof will use the following result on the geometry of stabilizer states:
\begin{fact}[Corollary 3 of \cite{GMC17}]\label{fact:stab-geometry}
Let $\ket{\zeta},\ket{\xi}$ be two stabilizer states. If $\abs{\bra{\zeta}\ket{\xi}}\neq 1$, then $\abs{\bra{\zeta}\ket{\xi}}\leq\frac{1}{\sqrt{2}}$.
\end{fact}

\begin{proof}[Proof of Lemma \ref{lem:H-up}]
Since $\had$ is a Clifford gate, $\had^{\otimes k}\ket{\eta}$ is a stabilizer state. We will show that $\abs{\bra{\eta}\had^{\otimes k}\ket{\eta}}\neq 1$ in the case of odd $k$, which by Fact \ref{fact:stab-geometry} will imply the bound.

Recall that $\ketbra{\eta}=\frac{1}{|G|}\sum_{g\in G}g$, where $G\equiv\stab(\ket{\eta})$. We have two cases:

\begin{enumerate}[label=\textbf{\Roman*.}]
    \item (Every $S\in G$ contains an $\eye$ or a $Y$ in some position) In this case, we calculate
    \begin{align*}
        \bra{\eta}\had^{\otimes k}\ket{\eta} &= \Tr[\ketbra{\eta}\had^{\otimes k}], \\
                            &= \frac{1}{|G|}\sum_{g\in G}\Tr[g\had^{\otimes k}], \\
                            &= \frac{1}{|G|}\sum_{g\in G}\prod_{i\in [k]} \Tr[g_i\had], \\
                            &= 0,
    \end{align*}
    where the last line follows since $g_j\in\{\eye,Y\}$ for some $j$, and $\Tr[\had]=\Tr[Y\had]=0$.

    \item (There is an $S\in G$ which consists of only $X$'s and $Z$'s) Consider the case when $k$ is odd. Since $\wt(S)=k$, $S$ contains either (1) an odd number of $X$'s and an even number of $Z$'s or (2) an even number of $X$'s and an odd number of $Z$'s. We focus on the former situation; the latter is similar.
    
    Note that $\abs{\bra{\eta}\had^{\otimes k}\ket{\eta}}=1$ if and only if $\had^{\otimes k}\ket{\eta}$ and $\ket{\eta}$ have the same stabilizer group. Since $S$ stabilizes $\ket{\eta}$, $\had^{\otimes k} S \had^{\otimes k}$ stabilizes $\had^{\otimes k}\ket{\eta}$. We know how $\had$ conjugates Pauli operators: $X\mapsto Z$, $Z\mapsto X$, and $Y\mapsto -Y$. By assumption, $S$ has an odd number of $X$'s and an even number of $Z$'s, so $\had^{\otimes k} S \had^{\otimes k}$ will have an even number of $X$'s and an odd number of $Z$'s. Therefore, we have that $S\cdot(\had^{\otimes k}S\had^{\otimes k})=-(\had^{\otimes k}S\had^{\otimes k})\cdot S$, which implies $S$ and $\had^{\otimes k}S\had^{\otimes k}$ cannot both be elements of the same stabilizer group. Hence, $\stab(\ket{\eta})\neq\stab(\had^{\otimes k}\ket{\eta})$ and $\abs{\bra{\eta}\had^{\otimes k}\ket{\eta}}\neq 1$.
\end{enumerate}
Since in both cases $\abs{\bra{\eta}\had^{\otimes k}\ket{\eta}}\neq 1$, by Fact \ref{fact:stab-geometry} we must have that $\abs{\bra{\eta}\had^{\otimes k}\ket{\eta}}\leq \frac{1}{\sqrt{2}}$ when $k$ is odd. We note that the above proof will not work for even $k$, since it can be the case that all stabilizers have an even number of $X$'s and $Z$'s (or both odd). In this case $\had^{\otimes k}$ will be in the normalizer of $G$, and the two stabilizer groups may be equal. 

We can easily find an example with even $k$ where no non-trivial upper bound can be found. Note that $\ket{\Phi^+}\equiv \frac{1}{\sqrt{2}}(\ket{00}+\ket{11})$ is a $+1$ eigenstate of $\had^{\otimes 2}$, so for even $k$ define $\ket{\eta_0}\equiv \ket{\Phi^+}^{\otimes k/2}$.
\end{proof}

We can now prove the local lower bound on odd-weight CSS generators.

\begin{lemma}\label{clm:lb-Xk-Zk}
    For every $k$-qubit stabilizer state, $\ket{\eta}$, $\bra{\eta}\Tilde{\Pi}_{\bar X}\ket{\eta}\geq c_k$ and $\bra{\eta}\Tilde{\Pi}_{\bar Z}\ket{\eta}\geq c_k$, where $c_k = 0$ if $k$ is even and $c_k = \sin^2(\frac{\pi}{8})$ if $k$ is odd.
\end{lemma}
\begin{proof}
Let $\ket{\eta}$ be a $k$-qubit stabilizer state. We first consider $\bra{\eta}\Tilde{\Pi}_{\bar X}\ket{\eta}$:
\begin{align}
    \bra{\eta}\Tilde{\Pi}_{\bar X}\ket{\eta} &\equiv \bra{\eta}D^{\dagger\otimes k}\left(\frac{\eye-\bar X}{2}\right)D^{\otimes k}\ket{\eta},\\
    \text{\color{since}(By Lemma \ref{lem:destabilize-XZ})}\hspace{3em}&= \bra{\eta}\frac{\eye-\had^{\otimes k}}{2}\ket{\eta} \\
    &=\frac{1}{2}\left(1-\bra{\eta} \had^{\otimes k}\ket{\eta}\right).
\end{align}
The bound follows from Lemma \ref{lem:H-up}, since $\sin^2(\frac{\pi}{8})=\frac{1}{2}(1-\frac{1}{\sqrt{2}})$.

For $\bra{\eta}\Tilde{\Pi}_{\bar Z}\ket{\eta}$, we have:
\begin{align}
    \bra{\eta}\Tilde{\Pi}_{\bar Z}\ket{\eta} &\equiv \bra{\eta}D^{\dagger\otimes k}\left(\frac{\eye-\bar Z}{2}\right)D^{\otimes k}\ket{\eta}, \\
    \text{\color{since}(By Lemma \ref{lem:destabilize-XZ})}\hspace{3em}&= \bra{\eta}\frac{\eye-(-X\had X)^{\otimes k}}{2}\ket{\eta} \\
    &=\frac{1}{2}\left(1-\bra{\eta} (-X\had X)^{\otimes k}\ket{\eta}\right),\\
    &=\frac{1}{2}\left(1-(-1)^k\bra{\zeta} \had^{\otimes k}\ket{\zeta}\right)
\end{align}
where $\ket{\zeta}\equiv X^{\otimes k}\ket{\eta}$ is another stabilizer state since $X=X^\dagger$ is in the Clifford group. The bound follows again from Lemma \ref{lem:H-up}.

\end{proof}

Lemma \ref{clm:lb-Xk-Zk} implies the following lower bound for $n$-qubit stabilizer states.

\begin{lemma}\label{lem:lb-stab-n-pure}
Let $S\in\mcP_n$ be a tensor product of only Pauli $X$ and $\eye$ or only Pauli $Z$ and $\eye$. Denote $k=\wt(S)$. For every $n$-qubit stabilizer state, $\ket{\eta}$, $\bra{\eta}\Tilde{\Pi}_S\ket{\eta}\geq c_k$.
\end{lemma}
\begin{proof}
Recall that $\Tilde{\Pi}_S = {\Tilde{\Pi}_S}\mid_{\N(S)}\otimes \eye_{[n]\setminus \N(S)}$, so
\begin{align}
    \bra{\eta}\Tilde{\Pi}_S\ket{\eta} &= \Tr[{\eta_{\N(S)} {\Tilde{\Pi}_S}\mid_{\N(S)}}],
\end{align}
where $\eta_{\N(S)}\equiv \Tr_{-{\N(S)}}[\ketbra{\eta}]$ is the reduced state of $\ket{\eta}$ on ${\N(S)}\subset [n]$. Since $\eta_{\N(S)}$ is the reduced state of a Clifford state, by Claim \ref{clm:Clifford-decomp} there are pure stabilizer states on $k$ qubits, $\{\ket{\eta_j}\}$ such that $\eta_{\N(S)} = \sum_j p_j\ketbra{\eta_j}$. The lower bound follows by applying Lemma \ref{clm:lb-Xk-Zk} to each $\bra{\eta_j} {\Tilde{\Pi}_S}\mid_{\N(S)}\ket{\eta_j}$.
\end{proof}

We can now prove Theorem \ref{thm:CSS-NLCS}.

\oddWeightCSS*
\begin{proof}
By definition, $\Tilde{\ham}_{\mcS_n} = \frac{1}{\abs{\mcS_n}}\sum_{S\in\mcS_n} \Tilde{\Pi}_S$ where $\Tilde{\Pi}_S$ is the $D$-rotated projector associated with $S\in\mcS_n$. Let $\psi$ be a stabilizer state on $n$ qubits. We will directly lower-bound the energy of $\psi$.

By definition, $\psi = \sum_j p_j\ketbra{\varphi_j}$, where each $\ket{\varphi_j}$ is a pure stabilizer state on $n$ qubits. We have:
\begin{align}
    \Tr[\psi \Tilde{\ham}_{\mcS_n}] &= \sum_j p_j \bra{\varphi_j}\Tilde{\ham}_{\mcS_n}\ket{\varphi_j}, \\
        &= \frac{1}{\abs{\mcS_n}}\sum_{S\in\mcS_n} \sum_j p_j \bra{\varphi_j}\Tilde{\Pi}_{S}\ket{\varphi_j}, \\
        \text{\color{since}(By Lemma \ref{lem:lb-stab-n-pure})\hspace{3em}}&\geq \frac{1}{\abs{\mcS_n}}\sum_{S\in\mcS_n} c_{\wt(S)} \sum_j p_j, \\
        \text{\color{since}(Definition of $c_k$)\hspace{3em}}&= \frac{1}{\abs{\mcS_n}}\sum_{S\in\mcS_n: \wt(S)\text{, odd}} \sin^2\Big(\frac{\pi}{8}\Big), \\
        &= \alpha \sin^2\Big(\frac{\pi}{8}\Big),
\end{align}
where the last line follows by assumption $\alpha\abs{\mcS_n}$ terms of $\mcS_n$ have odd weight. Since this holds for all stabilizer states, $\psi$, we have that $\{\Tilde{\ham}_{\mcS_n}\}$ is $\epsilon$-NLCS with $\epsilon=\alpha\sin^2(\frac{\pi}{8})=\Omega(1)$.
\end{proof}

We now turn to our main result, the existence of a simultaneous NLTS and NLCS family of local Hamiltonians. Recall the NLTS result of \cite{ABN22}:
\begin{theorem*}[Theorem 5 of \cite{ABN22}, simplified]
There exists a constant $\epsilon_0>0$ and an explicit family of CSS Hamiltonians associated with a family of quantum LDPC codes, $\{\langle \mcS_n\rangle\}$, which is $\epsilon_0$-NLTS.
\end{theorem*}
In order to use our Theorem \ref{thm:CSS-NLCS}, we require that a constant fraction of the stabilizer generators in $\mcS_n$ have an odd weight. It is not immediately clear that this would be true for the quantum Tanner codes from \cite{LZ22}. However, Section \ref{app:odd-weight} is dedicated to proving the following result.

\begin{restatable}[]{claim}{oddWeight}\label{clm:NLTS-odd-weight}
There exists an explicit family of CSS codes satisfying the conditions of Theorem 5 of \cite{ABN22} such that every stabilizer-generator has odd weight.
\end{restatable}

With Claim \ref{clm:NLTS-odd-weight}, we are now prepared to prove the main result of our paper.

\begin{theorem}\label{thm:NLTS-NLCS}
Let $\{\Hn\}$ be the family of CSS Hamiltonians from Claim \ref{clm:NLTS-odd-weight}. The $D$-rotated version, $\{\dHn\}$, is a family of simultaneous NLTS and NLCS local Hamiltonians.
\end{theorem}
\begin{proof}
Since $\{\Hn\}$ satisfies the conditions of Theorem 5 of \cite{ABN22} it is a valid local Hamiltonian, and it is $\epsilon_0$-NLTS for some constant $\epsilon_0>0$. Since $D\n$ is a depth-$\mcO(1)$ circuit by Lemma \ref{lem:NLTS-phase} the $D$-rotated family $\{\dHn\}$ is also $\epsilon_0$-NLTS. 

By Claim \ref{clm:NLTS-odd-weight}, all of the stabilizer terms of $\Hn$ have odd weight for every $n$. Thus, by Theorem \ref{thm:CSS-NLCS} $\{\dHn\}$ is $\epsilon_1$-NLCS for $\epsilon_1\equiv\sin^2(\frac{\pi}{8})$. Letting $\epsilon'\equiv\min\{\epsilon_0,\epsilon_1\}$, we have that $\{\dHn\}$ is both $\epsilon'$-NLTS and $\epsilon'$-NLCS.
\end{proof}


\section{Classical and quantum Tanner codes}\label{app:odd-weight}
All known examples of asymptotically good (distance $\Omega(n)$ and rate $\Omega(1)$) quantum LDPC codes are CSS codes \cite{BE21, PK22,LZ22}. As in the recent proof of NLTS in \cite{ABN22}, we will consider the quantum Tanner codes of \cite{LZ22}. The only condition required in Theorem \ref{thm:CSS-NLCS} is that a constant fraction of the CSS generators have an odd weight. As quantum Tanner codes have a random local code component, we will demonstrate here that the local code can be chosen in such a way that all of the stabilizers have odd weight (Claim \ref{clm:NLTS-odd-weight}). Since we will not need any results about the distance or rate of the code we choose we will not mention many of the intricate details required to show quantum Tanner codes are good quantum LDPC codes. Instead, we will simply mention the classical coding theory needed within their construction. We begin with some review.

\begin{definition*}
A \textbf{classical linear code} on $n$ bits is a subspace $\mcC\subseteq\FF_2^n$. Given such a linear code there is an integer $m\leq n$ and two full-rank matrices $H\in\FF_2^{m\times n}, G\in\FF_2^{(n-m)\times n}$ such that $\ker H = \im G^T = \mcC$. $H$ is called a \textbf{parity-check matrix} of $\mcC$ and $G$ is called a \textbf{generator matrix} of $\mcC$. In particular, $HG^T=0$ and $x\in\mcC$ if and only if $Hx=0$, the zero vector. Further, any $H\in\FF_2^{m\times n}$ defines a valid code, $\ker H$, and we can always find a $G\in\FF_2^{(n-\text{rk}(H))\times n}$ with $\im G^T = \ker H$. We will often work directly with $H$ instead of the underlying code, and indeed many classical codes are constructed by first constructing a parity-check matrix.
\end{definition*}
We mention that although we typically say \textit{the} parity-check/generator matrix, the choices of parity-check and generator matrices are not unique. Also, we need not restrict ourselves to full-rank matrices when defining codes.
\begin{definition*}
Given a linear code, $\mcC$, its \textbf{dual code} is defined as
\begin{equation*}
    \mcC^\perp \equiv \left\{ x\in\FF_2^n \mid x^T c = 0 \text{ for all }c\in\mcC \right\},
\end{equation*}
i.e. the dual code is the orthogonal complement of $\mcC$ under the standard dot product. Note that if $H$ (resp. $G$) is a parity-check (resp. generator) matrix of $\mcC$, then $H$ is a generator matrix (resp. parity-check) matrix of the dual code of $\mcC$, and vice versa.
\end{definition*}
As the name suggests, a major component of quantum Tanner codes is the notion of a classical Tanner code. Tanner codes are a way to build a larger code (known as the global code) from a smaller code (known as the local code) in such a way that properties of the local code translate to desirable properties of the global code. 
\begin{definition*}
Let $\mcG=(V,E)$ be a $d$-regular graph with $\abs{E}=n$. For a vertex $v\in V$ we denote the $d$ edges connected to $v$ as $E_v$.\footnote{We also assume that the edges in $E$ and in $E_v$ for each $v$ have been given a well-defined ordering.} Let $\local\subseteq \FF_2^d$ be a code on $d$ bits. The \textbf{Tanner code} associated with $\mcG$ and $\mcC$ is defined as
\begin{equation*}
    \mcT(\mcG,\local) \equiv \left\{ x\in\FF_2^n \;\Big|\; x_v\in\local \text{ for all }v\in V \right\},
\end{equation*}
where $x$ is some assignment of bits to $E$ and $x_v$ is the restriction of $x$ to the edges in $E_v$. We say that $\local$ is the \textbf{local code} of $\mcT(\mcG,\local)$. If $h\in\FF_2^{r\times d}$ is a parity-check matrix for $\mcC$, we will typically denote $\mcT(\mcG,\mcC)=\mcT(\mcG,\ker h)$. We can construct a  ``global'' parity-check matrix, $H$, for $\mcT(\mcG,\ker h)$ from $\mcG$ and $h$ as follows: $H$ has $r\abs{V}$ total rows, $r$ for each vertex of $\mcG$, and $\abs{E}$ columns, one for each edge of $\mcG$. Given a vertex, $v$, and a row $h_j$ of the local code $h$, the row $H_{v,j}$ of $H$ is defined as $H_{v,j}|_{E_v}=h_j$, and 0 elsewhere. In other words, if we restrict to the columns of $H_{v,j}$ corresponding to the edges connected to $v$, we will see the $j$th parity-check of $h$.
\end{definition*}
Intuitively, in a Tanner code we assign bits to the edges of $\mcG$, and for each vertex of $\mcG$ we check that the ``local'' view of each vertex is an element of the chosen local code. The distinction between local and global codes is crucial, and the notation use becomes heavy especially when discussing quantum Tanner codes. For the sake of consistency, we will always use calligraphic $\mcC$ and uppercase $H,G$ when referring to a global code and global parity-check or generator matrices. We will always use the standard uppercase $\local$ and lowercase $h,g$ when referring to a local code and local parity-check or generator matrices.

\begin{example}
    Consider Figure \ref{fig:tanner-code}, and let $x\in\FF_2^n$ be the assignment of bits to edges as shown. Suppose we choose a local code, $\ker h$, via the following parity-check matrix:
    \begin{equation*}
        h \equiv \begin{bmatrix} 1 & 0 & 0 \\ 0 & 1 & 1 \end{bmatrix}.
    \end{equation*}
    $x$ is not an element of the Tanner code $\mcT(\mcG,h)$. It passes the local check at vertex $v$ since $hx_v=[0\;\;0]^T$, but it does not pass the local check at vertex $w$ since $hx_w = [1\;\;1]^T$.
    
    \begin{figure}[H]
        \centering
        \includegraphics{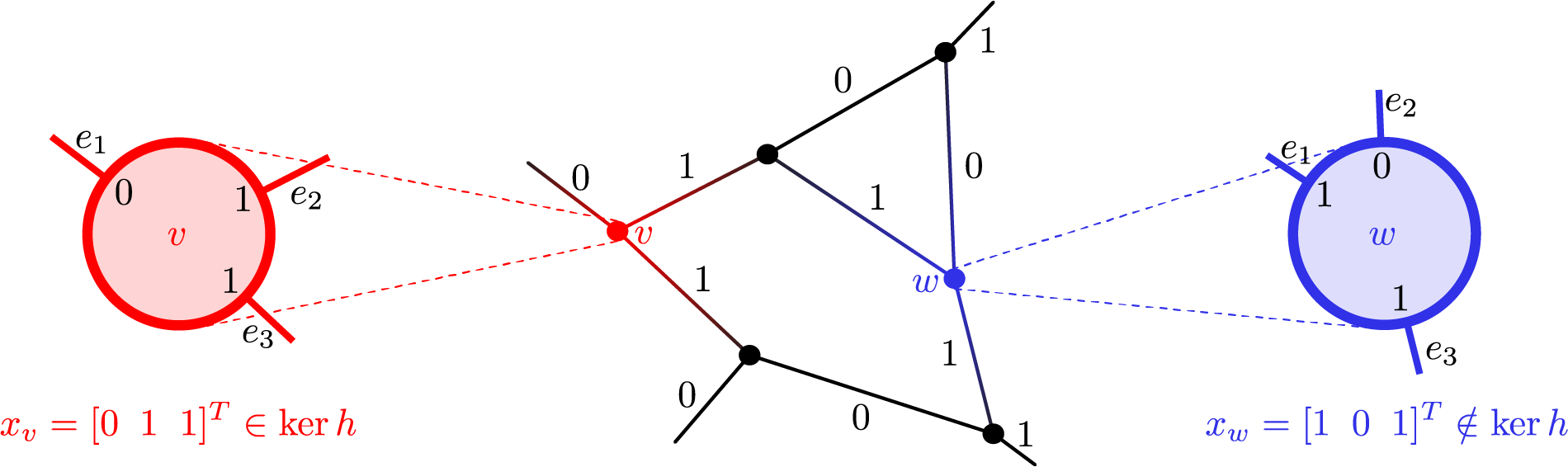}
        \caption{Tanner code example. Note that this is just a portion of the overall graph.}
        \label{fig:tanner-code}
    \end{figure}
\end{example}
Sipser and Spielman \cite{SS96} gave the first explicit example of asymptotically good classical LDPC codes by defining a Tanner code on an expanding, $d$-regular base graph\footnote{We really mean a countable family of graphs $\{\mcG_n\}$ with constant degree $d$.}, $\mcG$. In their original result the local code was given by a uniformly random parity-check $h\in\FF_2^{r\times d}$, and they showed that the resulting code $\mcT(\mcG,\ker h)$ has asymptotically good parameters with high probability.

The definition of a CSS code may seem starkly different from that of a classical code, and it may not be clear how classical codes could be useful in the construction of quantum codes. However, the following more coding-theoretic definition of a CSS code is equivalent to the definition given in Section \ref{sec:prelim}.
\begin{definition*}[Alternate definition of a CSS code]
Let $\mcC_X$ and $\mcC_Z$ be two classical codes with parity-check matrices $H_X\in\FF_2^{m_x\times n}$ and $H_Z\in\FF_2^{m_z\times n}$. Suppose $H_X$ and $H_Z$ satisfy $H_X H_Z^T=0$, i.e. $\mcC_Z^\perp\subseteq \mcC_X$. For each row of $H_X$ (resp. $H_Z$) we associate an element of the Pauli group, $S_1\otimes \dots \otimes S_n\in\mcP_n$, where $S_j=X$ (resp. $S_j=Z$) if the $j$th element is 1, and $S_j=\eye$ if the $j$th element is 0. The set of these Pauli elements forms a generating set for a valid stabilizer group and thus define a CSS code which we denote $\CSS(\mcC_X,\mcC_Z)$.
\end{definition*}

In the same spirit as \cite{SS96}, Leverrier and Zémor \cite{LZ22} showed for an explicit $d^2$-regular graph\footnote{Technically, they define two families of graphs $\{\mcG_{0,n}^\square\}$ and $\{\mcG_{1,n}^\square\}$ which are defined on the same vertex set, but have different edge sets which are related via a ``square complex'' structure. Our argument will only rely on their vertex sets having the same size, so for simplicity we will suppose there is a single graph, $\mcG^\square$.}, $\mcG^\square$, and two special local codes, $\local_X$ and $\local_Z$, that the global codes $\mcC_X\equiv\mcT(\mcG^\square,\local_X)$ and $\mcC_Z\equiv\mcT(\mcG^\square,\local_Z)$ define an asymptotically good CSS code with high probability. Anshu, Breuckmann, and Nirkhe \cite{ABN22} later showed that this CSS code satisfies a sufficient property for its associated local Hamiltonian to be NLTS.

By definition, the code $\CSS(\mcC_X,\mcC_Z)$ will satisfy the condition of Theorem \ref{thm:CSS-NLCS} if a constant fraction of the rows of both $H_X$ and $H_Z$ have odd weight. Thus, we need to show that we can choose global parity-check matrices of $\mcC_X$ and $\mcC_Z$ such that, for each, a constant fraction of the rows have odd weight. We can do this by looking at the local parity-check matrices and using the following:
\begin{lemma}\label{lem:odd-local-implies-odd-global}
Let $\mcT(\mcG,\ker h)$ be a Tanner code. If an $\alpha>0$ fraction of the rows of $h\in\FF_2^{r\times d}$ have odd weight, then an $\alpha$ fraction of the rows of the corresponding global parity-check matrix, $H$, have odd weight, as well. 
\end{lemma}
\begin{proof}
By assumption, for every $v\in V$, $\alpha r$ of the $r$ rows of $H$ corresponding to $v$ have odd weight. Thus $\alpha r \abs{V}$ of the $r\abs{V}$ rows of $H$ have odd weight and the implication holds. 
Note that since the dimensions of $h$ are both fixed constants, if even a single row of the local parity-check has odd weight then at least a $1/r>0$ fraction of the rows for the global parity-check also have odd weight, which is a constant.
\end{proof}
This shows, for instance, that half of the rows of a classical expander code have odd weight since the local parity-check matrix is chosen uniformly at random. If the local codes, $\mcC_X,\mcC_Z$, of \cite{LZ22} were also chosen via uniformly random parity-check matrices we would be done. Unfortunately, the local codes of a quantum Tanner code are more structured.

Given two codes, $\local_0,\local_1\subseteq\FF_2^d$, on $d$ bits, their ``dual tensor code'' is a code on $d^2$ bits defined as 
\begin{equation*}
    \local_0\boxplus \local_1\equiv \Big(\local_0^\perp \otimes \local_1^\perp\Big)^\perp.
\end{equation*}
The definition is rather opaque, but there is an easier way to view the dual tensor code:
\begin{fact}\label{fact:parity-check-of-dual-tensor}
    If $h_0, h_1$ are parity-check matrices for $\local_0, \local_1$, respectively, then
    \begin{equation*}
        \ker\big(h_0 \otimes h_1\big) = \local_0\boxplus \local_1.
    \end{equation*}
    That is, $h_0 \otimes h_1$ is a parity-check matrix for the dual tensor code of $\local_0$ and $\local_1$.
\end{fact}

The local codes of the constituent \textit{classical} Tanner codes of a quantum Tanner code are both dual tensor codes. We now state the main result of \cite{LZ22} on quantum LDPC codes. Note that we are combining two theorems to better fit our discussion.
\begin{fact}[Theorem 2 and Theorem 17 of \cite{LZ22}]\label{fact:LZ}
Take $\mcG^\square$ to be the family of $d^2$-regular graphs constructed in \cite{LZ22} and $r \equiv \floor{\rho d}$ for $\rho\in(0,1/2)$. Choose $h_0, h_1\in\FF_2^{r\times d}$ by picking entries uniformly at random. Let $h_0$ be the parity-check matrix for the code $\local_0$ and $h_1$ be the generator matrix for the code $\local_1$, i.e. $\local_0=\ker h_0$ and $\local_1=\im h_1^T$. With probability 1 when $d$ goes to infinity, the classical Tanner codes $\mcC_X\equiv \mcT(\mcG^\square,\local_0\boxplus\local_1)$ and $\mcC_Z\equiv\mcT(\mcG^\square,\local_0^\perp\boxplus\local_1^\perp)$ yield a family of asymptotically good quantum LDPC codes.
\end{fact}

As mentioned previously, the local codes for the two classical codes are, themselves, constructed from two smaller codes. The reason for choosing a random parity-check matrix for $\local_0$ and a random generator matrix for $\local_1$ is to provide a symmetry in the construction, particularly because $h_0$ is a random \textit{generator} matrix for $\local_0^\perp$ and $h_1$ is a random \textit{parity-check} matrix for $\local_1^\perp$.

To utilize Lemma \ref{lem:odd-local-implies-odd-global} we will first need local parity-check matrices for $\local_0\boxplus \local_1$ and $\local_0^\perp \boxplus \local_1^\perp$. By Fact \ref{fact:parity-check-of-dual-tensor}, this is equivalent to choosing parity-check matrices for $\local_0,\local_1,\local_0^\perp,$ and $\local_1^\perp$. We already have parity-check matrices for $\local_0$ and $\local_1^\perp$ from the statement of Fact \ref{fact:LZ}, so we must choose parity-check matrices for $\local_1$ and $\local_0^\perp$.

We will need the following technical lemma to prove the main result of this section:
\begin{lemma}\label{lem:single-odd-weights}
    Choose $A\in\FF_2^{r\times d}$ by picking entries in $\{0,1\}$ i.i.d. uniformly at random, where $r=\floor{\rho d}$ for $\rho\in(0,1/2)$. With probability 1 when $d$ goes to infinity, the following hold:
    \begin{enumerate}
        \item At least one row of $A$ has odd weight.
        \item At least one $x\in\ker A\subseteq\FF_2^d$ has odd weight.
    \end{enumerate}
\end{lemma}
\begin{proof}
We denote the $i$-th row of $A$ by $A_i$.

(1.) Let $\evenodd(v)\equiv \wt(v)\mod{2}$ be the parity of $v\in\FF_2^d$, and let $O_A$ be the [non-negative] random variable corresponding to the number of odd weight rows in $A$, i.e. $O_A \equiv \abs{\{A_i\mid \evenodd(A_i)=1\}}$.
    \begin{align}
        \Pr[O_A>0] &= 1-\Pr[O_A=0], \\
            &= 1-\Pr\Big[\evenodd(A_i)=0 \;\; \forall i\in[r]\Big], \\
            {\color{since}(\because \text{i.i.d.})}\hspace{1em} &= 1-\Pr\Big[\evenodd(A_1)=0\Big]^r, \\
            &= 1-(1/2)^r,
    \end{align}
    where the last line follows since $A_1\in\FF_2^d$ is uniformly random and half of the vectors in $\FF_2^d$ have even weight. Since $r=\floor{\rho d}$ the value approaches 1 when $d$ goes to infinity.
    
 (2.) We will bound the probability that every $x\in\ker A$ has even weight. Every codeword has even weight if and only if the all-ones vector, $[1,1,\dots,1]$, is a parity-check of $\ker A$. That is, precisely when $1^d\equiv [1,1,\dots,1] \in\im A^T$, the row space of $A$. Let $S_k\equiv\Span\{A_1,\dots,A_k\}$ denote the span of the first $k$ rows of $A$ and note that $S_r=\im A^T$. We will prove by induction on $r$ that $\Pr[1^d\in S_r]\leq (2^r-1)/2^d$.

    When $r=1$, $\Pr[1^d\in S_1]=\Pr[A_1=1^d]=1/2^d$, as desired. For $k\geq 1$ suppose $\Pr[1^d\in S_k]\leq (2^k-1)/2^d$. Consider $k+1$. By the law of total probability we have
    \begin{align}
        \Pr[1^d\in S_{k+1}] = &\Pr[1^d\in S_{k+1}\;\Big|\; 1^d\in S_{k}]\cdot \Pr[1^d\in S_{k}] \nonumber\\
        &+ \Pr[1^d\in S_{k+1}\;\Big|\; 1^d\notin S_{k}]\cdot \Pr[1^d\notin S_{k}], \\
        {\color{since}(\Pr[1^d\notin S_{k}]\leq 1)}\hspace{1.5em} \leq &\Pr[1^d\in S_{k+1}\;\Big|\; 1^d\in S_{k}]\cdot \Pr[1^d\in S_{k}] \nonumber\\
        &+ \Pr[1^d\in S_{k+1}\;\Big|\; 1^d\notin S_{k}], \\
        {\color{since}(1^d\in S_k\Rightarrow 1^d\in S_{k+1})}\hspace{1.5em} = &\Pr[1^d\in S_{k}] + \Pr[1^d\in S_{k+1}\;\Big|\; 1^d\notin S_{k}], \nonumber \\
        {\color{since}(\text{I.H.})}\hspace{1.5em} \leq &\frac{2^k-1}{2^d} + \Pr[1^d\in S_{k+1}\;\Big|\; 1^d\notin S_{k}].
    \end{align}
    Now consider if $1^d\notin S_k=\Span\{A_1,\dots,A_k\}$. Since $\abs{S_k}\leq2^k$, $1^d\in S_{k+1}$ only if $A_k=1^d+v$ for some $v\in S_k$. Since $A_k$ is chosen uniformly at random, this occurs with probability $2^k/2^d$. Thus,
    \begin{equation}
        \Pr[1^d\in S_{k+1}] \leq \frac{2^k-1}{2^d} + \frac{2^k}{2^d} = \frac{2^{k+1}-1}{2^d},
    \end{equation}
    as desired. Altogether we have
    \begin{equation}
        \Pr[\text{\# odd weight vectors in }\ker A >0] = 1 - \Pr[1^d\in S_r]\geq 1-\frac{2^{\floor{\rho d}}-1}{2^d},
    \end{equation}
    which goes to 1 when $d$ goes to infinity.

\end{proof}

\begin{corollary}\label{cor:single-odd-parity-checks}
    Let $h$ be chosen as in Lemma \ref{lem:single-odd-weights} and suppose $\local \equiv\ker h$. With probability 1 when $d$ goes to infinity every parity-check matrix, $g\in\FF_2^{(d-r)\times d}$, for $\local^\perp$ has at least one row with odd weight.
\end{corollary}
\begin{proof}
    Recall that a parity-check matrix for $\local^\perp$ is a generator matrix for $\local$, and hence the rows of $g$ generate the codewords in $\local=\ker h$. Suppose for contradiction that all of the rows of $g$ have even weight. The sum of any two even-weight vectors over $\FF_2$  has even weight, so it must be that all of the codewords of $\local$ have even weight. However, by point 2 of Lemma \ref{lem:single-odd-weights} at least one $x\in\local$ has odd weight, a contradiction.
\end{proof}

We now prove the main result of this section, Claim \ref{clm:NLTS-odd-weight}. 
\oddWeight*

\begin{proof}
The family of CSS codes used in Theorem 5 of \cite{ABN22} is precisely the family $\CSS(\mcC_X,\mcC_Z)$ from Fact \ref{fact:LZ} given by $\mcC_X\equiv \mcT(\mcG^\square,\local_0\boxplus\local_1)$ and $\mcC_Z\equiv\mcT(\mcG^\square,\local_0^\perp\boxplus\local_1^\perp)$. We must show that the global parity-check matrices, $H_X$ and $H_Z$, of $\mcC_X$ and $\mcC_Z$ can be chosen to have only odd weight rows.

By Fact \ref{fact:LZ}, with probability 1 when $d$ goes to infinity, the random choice of a parity-check matrix, $h_0\in\FF_2^{r\times d}$, for $\local_0$ and of a generator matrix, $h_1\in\FF_2^{r\times d}$, for $\local_1$, yields asymptotically good quantum LDPC codes. By Fact \ref{fact:parity-check-of-dual-tensor}, a parity-check matrix of $\local_0\boxplus \local_1$ is given by $h_0\otimes g_1$, where $g_1$ is any parity-check matrix of $\local_1=(\local_1^\perp)^\perp$, and a parity-check matrix of $\local_0^\perp\boxplus \local_1^\perp$ is given by $g_0\otimes h_1$, where $g_0$ is any parity-check matrix of $\local_0^\perp$. 

By point 1 of Lemma \ref{lem:single-odd-weights}, with probability 1 when $d$ goes to infinity both $h_0$ and $h_1$ have at least one row with odd weight. Similarly, by Corollary \ref{cor:single-odd-parity-checks}, with probability 1 when $d$ goes to infinity both $g_0$ and $g_1$ have at least one row with odd weight. 
We construct a new matrix, $h_0'$, as follows: let $h_{0,\star}$ be an odd-weight row of $h_0$. We let each odd-weight row of $h_0$ be a row in $h_0'$, and for each even-weight row, $h_{0,j}$, there is a row $h_{0,j}+h_{0,\star}$ in $h_0'$. Since the sum of an even-weight vector and an odd-weight vector has odd weight, the matrix $h_0'$ has only odd-weight rows (given $h_0$ has at least one odd-weight row). We also note that adding one row to another does not change the kernel of a matrix. That is $\ker h_0 = \ker h_0'$. We construct matrices $h_1',g_0',$ and $g_1'$ analogously. By construction, all of the rows of $h_0',h_1',g_0',$ and $g_1'$ have odd weight.

Since $\ker( h_0'\otimes g_1' )=\ker( h_0\otimes g_1 )=\local_0\boxplus\local_1$, $h_0'\otimes g_1'$ is a parity-check matrix for the local code of $\mcC_X\equiv \mcT(\mcG^\square,\local_0\boxplus\local_1)$ used by Theorem 5 of \cite{ABN22}. As every row of $h_0'$ and $g_0'$ has odd weight, every row of their tensor product does, as well. Thus, by Lemma \ref{lem:odd-local-implies-odd-global} every row of the corresponding global parity-check matrix, $H_X$, constructed from $h_0'$ and $g_1'$ has odd weight. The same is true for the global parity-check matrix, $H_Z$, constructed from $g_0'$ and $h_1'$.

As all of the properties hold with probability 1 when $d$ goes to infinity, there is some constant $d$ such that there are explicit matrices $h_0',h_1',g_0',$ and $g_1'$ satisfying all of the conditions of this claim, Fact \ref{fact:LZ} (Theorem 2 and Theorem 17 of \cite{LZ22}), and Theorem 5 of \cite{ABN22}.
\end{proof}


\section{Future work}
\begin{enumerate}[label=(\arabic*), leftmargin=*]
\item The most immediate problem raised by this work is to show that rotating arbitrary CSS Hamiltonians by $\destabn$ yields NLCS Hamiltonians. We have shown this when a constant fraction of the stabilizer generators have odd weight, which is a technical requirement of our proof technique. Nonetheless, we believe all \destab-rotated CSS Hamiltonians are NLCS. A first step would be to show this for $\ham\equiv \frac{1}{n}\sum \ketbra{11}_{i,i+1} = \frac{1}{n}\sum \frac{1}{2}(\eye-Z_iZ_{i+1})$, which has only even weight stabilizer generators.

\item NLACS Hamiltonians are an implication of either NLSS or the quantum PCP conjecture together with $\NP\neq\QMA$ (see Diagram \ref{diagram:qpcp-implications}), so we believe they exist. In Appendix \ref{app:NLCS} we give a self-contained proof that the simple $D$-rotated zero Hamiltonian, $\Tilde{\ham}_0 = \frac{1}{n}\sum(\destabdagger\ketbra{1}\destab)_i$, is NLCS, and in Appendix \ref{app:NLACS-zero}, we give a sharp lower-bound on the energy of states produced by Clifford + 1 $\T$ gate under $\Tilde{\ham}_0$. We also conjecture a sharp lower-bound on the energy for states prepared by Clifford + $t$ $\T$ gates, for any $t\leq n$.

\item We hope that our techniques may lead to local Hamiltonians which satisfy NLSS. Consider the zero Hamiltonian, $\ham_0 = \frac{1}{n}\sum\ketbra{1}_i$, and a family of Haar-random low-depth circuits, $C=\{C_n\}$. The unique ground-state of the local Hamiltonian $C\ham_0 C^\dagger$ is exactly $C\ket{0^n}$,\footnote{Note that we typically denote rotating by $C$ as $C^\dagger \ham C$, not  $C\ham C^\dagger$. We have swapped the order here so that the ground state is $C\ket{0^n}$, as opposed to $C^\dagger\ket{0^n}$.} which is not sampleable (as defined in Section \ref{sec:prelim}) unless $\P=\sharpP$ \cite{BFC+18, Mov20}. We hope that the same is true for states of low-enough constant energy, but new techniques would be necessary to show this. If true, $C\ham_0 C^\dagger$ would be an NLSS Hamiltonian unless $\P=\sharpP$.

Analogously to our result for simultaneous NLTS and NLCS, one may hope that rotating arbitrary CSS Hamiltonians by random low-depth circuits could also yield simultaneous NLTS and NLSS. However, there are many unresolved prerequisites needed to show this. For example, for a CSS Hamiltonian, $\ham$, every ground-state of $C\ham C^\dagger$ has the form $C\ket\psi$ for a codestate $\ket\psi$. It is not a fortiori true that applying a random low-depth circuit to codestates of a CSS code will result in a state that is not sampleable, so it is not clear that even the ground-space of such a Hamiltonian is not sampleable.

\item It is important to note that the technique of rotating Hamiltonians by a constant-depth circuit, while potentially useful for establishing NLSS, seemingly cannot provide certain other prerequisites of the quantum PCP conjecture.  For example, Fact \ref{fact:classical-local-access} says that the energies of locally-approximable states can be computed in $\NP$, and so the quantum PCP conjecture implies the following (assuming $\NP\neq\QMA$):
\begin{conjecture}[No Low-energy Locally-approximable States (NLLS)]
    There exists a family of local Hamiltonians, $\Hn$, and a constant $\epsilon>0$ such that all $\epsilon$-low-energy states of $\Hn$ are not locally-approximable.
\end{conjecture}
A closely-related conjecture (``no low-lying classically-evaluatable states'' conjecture) was very recently stated in \cite{WFC23}.\footnote{Note that these conjectures would not imply LH-$\epsilon\notin\NP$ as it would not rule out Hamiltonians whose ground-state energies have indirect $\NP$-witnesses. \cite{BV04} constructs such witnesses for certain commuting Hamiltonians.} Rotating by a constant-depth circuit preserves the NLLS property in the same way that it preserves the NLTS property, thus ruling out the use of rotating Hamiltonians in solving the NLLS conjecture. 

Furthermore, for any CSS Hamiltonian rotated by a constant-depth circuit, which includes every construction considered in this paper, the local Hamiltonian problem is contained in NP. 
To see this, note that every $C$-rotated CSS Hamiltonian has a ground state of the form $C^\dagger\ket\varphi$ for some stabilizer state $\ket\varphi$. Such states are locally-approximable since the local density matrices can be efficiently calculated by using a combination of the local density matrix calculation for trivial states and stabilizer states.

\end{enumerate}

\newpage
\bibliography{ref}
\newpage
\appendix

\section{Omitted proofs}
\subsection{Mixed Clifford states}\label{app:Clifford-decomp}

\begin{restatable}{definition}{stabsubgroups}\label{def:stab-subgroups}
Let $G$ be a stabilizer group, $P=P_1\otimes\dots\otimes P_n\in\mcP_n$ be any Pauli operator, and $A\subseteq[n]$ be any subset of $n$ qubits. We define the set $G_{A,P}$ to be
\begin{equation*}
    G_{A,P} \equiv \left\{  g_A \;\;\Big\vert\;\; g\in G,\; g_j=P_j \text{ for all }j\notin A  \right\},
\end{equation*}
where $g_A$ denote the restriction of $g$ to $A$ (note that $g_A$ acts on $\abs{A}$ qubits, not $n$ qubits).
\end{restatable}
$G_{A,P}$ can be thought of as all of the elements of $G$ which are equal to $P$ outside of the subset $A$, though we consider the restriction of these elements to $A$ only (including overall phases).
By abuse of notation we will denote $G_{i,P}\equiv G_{\{i\},P}$ and $ G_{-A,P}\equiv G_{[n]\setminus A,P}$ for $i\in[n]$.  We denote the special case of $G_{A,\eye}$ by $G_A$. $G_A \equiv \left\{ g_A \mid g\in G \text{ and } \N(g)\subseteq A\right\}\cup\{\eye_A\}$ is the set of all elements in $G$ which act non-trivially only on qubits in $A$.

Claim \ref{clm:Clifford-decomp} is immediate from the following two well-known facts.
\begin{fact}\label{fact:groundstateproj}
Let $G\leq\mcP_n$ be a stabilizer group and $\mcC$ the associated codespace. $\frac{1}{|G|}\sum_{g\in G}g$ is the projector onto $\mcC$. If $|G|=2^n$, then $\frac{1}{2^n}\sum_{g\in G}g =\ketbra{\psi}$, where $\ket{\psi}$ is the stabilizer state associated with $G$. Otherwise, $|G|=2^{n-r}$ for $r>0$ and there are $2^{r}$ logical basis states of $\mcC$. Let $\{\ket{\bar x}\}$ denote the logical computational basis states for $\mcC$. Then 
\begin{equation*}
    \frac{1}{2^{n-r}}\sum_{g\in G} g =\sum_{x\in\FF_2^{r}}\ketbra{\bar x}.
\end{equation*}
\end{fact} 

\begin{fact}\label{fact:mixed-Cliff-decomp}
Suppose $\ket{\psi}$ is a stabilizer state on $N$ qubits with stabilizer group $G$ and let $A$ be a subset of the qubits of size $n$. By Fact \ref{fact:groundstateproj} we can write $\ketbra{\psi} = \frac{1}{2^N}\sum_{g\in G}g$. 
The local state on $A$, $\psi\equiv\Tr_{-A}[\ketbra{\psi}]$, is equal to
\begin{equation*}
    \psi = \frac{1}{2^n}\sum_{\hat g \in G_A} \hat g.
\end{equation*}
\end{fact}

\Cliffdecomp*

\begin{proof}

By definition, there is a pure Clifford state $\ket{\psi}$ on $N\geq n$ qubits and a subset $A$ of $n$ qubits such that $\psi = \Tr_{-A}[\ketbra{\psi}]$. Let $G\equiv\stab(\ket{\psi})$, and let $G_A$ be defined as in Fact \ref{fact:mixed-Cliff-decomp}. By definition, $G_A$ is an abelian subgroup of $\mcP_n$ not containing $-\eye$, and so it is a valid stabilizer group. Let $|G_A |=2^{n-r}$. We have
\begin{align}
    \text{\color{since}(By Fact \ref{fact:mixed-Cliff-decomp})}\hspace{2.3em} \psi &=\frac{1}{2^r2^{n-r}}\sum_{\hat g\in G_A} \hat g, \\
    \text{\color{since}(By Fact \ref{fact:groundstateproj})}\hspace{3em} &= \frac{1}{2^{r}}\sum_{x\in\FF_2^{r}}\ketbra{\bar x}.
\end{align}
Since each $\ketbra{\bar x}$ is a stabilizer state on $n$ qubits and $\sum_{x\in\FF_2^{r}} \frac{1}{2^{r}} = 1$, the statement is proven.
\end{proof}

\newpage
\subsection{Rotated projectors}\label{app:lb-Xk-Zk}
Return to Claim \ref{lem:destabilize-XZ}.
\destabXZ*
\begin{proof}
We will show that $D^\dagger X D = \had$ and $D^\dagger Z D = -X\had X$. As ${\Pi}_{\bar X}\equiv (\frac{1}{2})(\eye-X)$ and  ${\Pi}_{\bar Z}\equiv (\frac{1}{2})(\eye-Z)$, the result follows.

\begin{align*}
    D^\dagger X D &= \left(\cos(\frac{\pi}{8})\eye+\sin(\frac{\pi}{8})ZX\right)X\left(\cos(\frac{\pi}{8})\eye+\sin(\frac{\pi}{8})XZ\right), \\
    &= \left(\cos(\frac{\pi}{8})X+\sin(\frac{\pi}{8})Z\right)\left(\cos(\frac{\pi}{8})\eye+\sin(\frac{\pi}{8})XZ\right), \\
    &= \cos^2\left(\frac{\pi}{8}\right)X+\sin(\frac{\pi}{8})\cos(\frac{\pi}{8})Z+\sin(\frac{\pi}{8})\cos(\frac{\pi}{8})Z-\sin^2\Big(\frac{\pi}{8}\Big)X, \\
    &= \cos(\frac{\pi}{4})Z+\sin(\frac{\pi}{4})X,\\
    &= \frac{1}{\sqrt{2}}(Z+X),\\
    &= \had.
\end{align*}

\begin{align*}
    D^\dagger Z D &= \left(\cos(\frac{\pi}{8})\eye+\sin(\frac{\pi}{8})ZX\right)Z\left(\cos(\frac{\pi}{8})\eye+\sin(\frac{\pi}{8})XZ\right), \\
    &= \left(\cos(\frac{\pi}{8})Z-\sin(\frac{\pi}{8})X\right)\left(\cos(\frac{\pi}{8})\eye+\sin(\frac{\pi}{8})XZ\right), \\
    &= \cos^2\Big(\frac{\pi}{8}\Big)Z-\sin(\frac{\pi}{8})\cos(\frac{\pi}{8})X-\sin(\frac{\pi}{8})\cos(\frac{\pi}{8})X-\sin^2\Big(\frac{\pi}{8}\Big)Z, \\
    &= \cos(\frac{\pi}{4})Z-\sin(\frac{\pi}{4})X,\\
    &= \frac{1}{\sqrt{2}}(Z-X),\\
    &= -X\had X.
\end{align*}
\end{proof}


\newpage
\section{A simple NLCS Hamiltonian}\label{app:NLCS}
The goal of this section is to demonstrate the existence of a simple family of NLCS Hamiltonians.

\begin{definition}\label{def:zero-Ham}
The \textbf{zero Hamiltonian}, $\Hzero$ is defined as
\begin{equation*}
    \Hzero \equiv \frac{1}{n}\sum_{i=1}^n \ketbra{1}_i\otimes \eye_{[n]\setminus\{i\}}.
\end{equation*}
Note that $\Hzero\ket{x}=\frac{\abs{x}}{n}\ket{x}$ for all $x\in\FF_2^n$. In particular, the unique ground state of $\Hzero$ is $\ketzero$ with energy $0$. For $n=1$ we have $\ham_0^{(1)} = \ketbra{1}$, so we can write the zero Hamiltonian on $n$ qubits as
\begin{equation*}
    \Hzero \equiv \frac{1}{n}\sum_{i=1}^n \ham_0^{(1)} \otimes \eye_{[n]\setminus\{i\}}.
\end{equation*}
\end{definition}
\begin{remark*}
    Define a set of stabilizer generators, $\mcS_n\equiv\{Z_1,\dots,Z_n\}$ where $Z_i$ is a Pauli $Z$ on qubit $i$ and identity elsewhere. The zero Hamiltonian is the CSS Hamiltonian associated with $\langle\mcS_n\rangle$, since $\ketbra{1}=\frac{\eye-Z}{2}$. The results of this section are a direct corollary of the results in Section \ref{sec:NLCS-CSS}.
\end{remark*}
Let $D\equiv\destab$. We define the $D$-rotated zero Hamiltonian as
\begin{equation*}
    \dHzero \equiv \frac{1}{n}\sum_{i=1}^n \Tilde{\ham}_0^{(1)} \otimes \eye_{[n]\setminus\{i\}},
\end{equation*}
where $\Tilde{\ham}_0^{(1)} = D^\dagger\ketbra{1}D$.
We will prove the $D$-rotated zero Hamiltonian is NLCS by demonstrating a simple lower bound on the energy of stabilizer states for each local term. Since the reduced state of every stabilizer state is a convex combination of stabilizer states by Claim \ref{clm:Clifford-decomp}, these ``local'' lower bounds imply a global lower bound for all stabilizer states. We have the following local energy bound.
Note that 

\begin{lemma}\label{lem:lb-zero-single-pure}
    If $\ket{\eta}$ is a single-qubit stabilizer state, then $\bra{\eta}{\Tilde{\ham}_0}^{(1)}\ket{\eta}\geq \sin^2(\frac{\pi}{8})$.
\end{lemma}

\begin{proof}
By definition $\Tilde{H}_0^{(1)}=D^\dagger \ketbra{1} D$, so $\bra{\eta}\Tilde{H}_0^{(1)}\ket{\eta} = \abs{\bra{1}D\ket{\eta}}^2$. As
\begin{equation*}
    D = \cos(\frac{\pi}{8})\eye-i\sin(\frac{\pi}{8})Y = \cos(\frac{\pi}{8})\eye+\sin(\frac{\pi}{8})XZ,
\end{equation*}
we have
\begin{align*}
    D = \begin{bmatrix}
    \cos(\frac{\pi}{8}) & -\sin(\frac{\pi}{8}) \\
    \sin(\frac{\pi}{8}) & \cos(\frac{\pi}{8})
    \end{bmatrix}.
\end{align*}
There are only six single-qubit stabilizer states to check: $\ket{0},\ket{1},\ket{+},\ket{-},\ket{y+},$ and $\ket{y-}$.
\begin{itemize}
    \item $\abs{\bra{1}D\ket{0}}^2=\sin^2(\frac{\pi}{8})$.
    \item $\abs{\bra{1}D\ket{1}}^2=\cos^2(\frac{\pi}{8})$.
    \item $\abs{\bra{1}D\ket{+}}^2=\frac{1}{2}(\sin(\frac{\pi}{8})+\cos(\frac{\pi}{8}))^2=\frac{1}{2}(1+\sin(\frac{\pi}{4}))=\cos^2(\frac{\pi}{8})$.
    \item $\abs{\bra{1}D\ket{-}}^2=\frac{1}{2}(\sin(\frac{\pi}{8})-\cos(\frac{\pi}{8}))^2=\frac{1}{2}(1-\sin(\frac{\pi}{4}))=\sin^2(\frac{\pi}{8})$.
    \item $\abs{\bra{1}D\ket{y+}}^2=\frac{1}{2}\abs{\sin(\frac{\pi}{8})+i\cos(\frac{\pi}{8})}^2=\frac{1}{2}$.
    \item $\abs{\bra{1}D\ket{y-}}^2=\frac{1}{2}\abs{\sin(\frac{\pi}{8})-i\cos(\frac{\pi}{8})}^2=\frac{1}{2}$.
\end{itemize}
Since $\sin^2(\frac{\pi}{8})$ is the smallest value, the result is proven.
\end{proof}

\begin{corollary}\label{cor:lb-zero-single-mixed}
    If $\eta$ is a single-qubit mixed stabilizer state, then $\Tr[\eta{\Tilde{\ham}_0}^{(1)}]\geq \sin^2(\frac{\pi}{8})$.
\end{corollary}
\begin{proof}
    By definition, $\eta = \sum_j p_j\ketbra{\varphi_j}$, where each $\ket{\varphi_j}$ is a pure stabilizer state on a single qubit. The lower bound follows by applying Lemma \ref{lem:lb-zero-single-pure} to each $\ket{\varphi_j}$.
\end{proof}

We now have the following global lower bound.
\begin{lemma}\label{lem:lb-zero-n-pure}
    If $\ket{\eta}$ is an $n$-qubit stabilizer state, then $\bra{\eta}\dHzero\ket{\eta}\geq \sin^2(\frac{\pi}{8})$.
\end{lemma}
\begin{proof}
By definition, $\dHzero =  \frac{1}{n}\sum_{i=1}^n {{\Tilde{\ham}_0}^{(1)}}\mid_i\otimes \eye_{[n]\setminus\{i\}}$, so
\begin{equation*}
    \bra{\eta}\dHzero\ket{\eta} =  \frac{1}{n}\sum_{i=1}^n \Tr[{\eta_i{\Tilde{\ham}_0}^{(1)}}],
\end{equation*}
where $\eta_i\equiv \Tr_{-i}[\ketbra{\eta}]$ is the reduced state of $\ket{\eta}$ on qubit $i$. Since $\eta_i$ is the reduced density matrix of a Clifford state, by Claim \ref{clm:Clifford-decomp} it is also a stabilizer state. The bound follows by applying Corollary \ref{cor:lb-zero-single-mixed} to each term in the summation.

\end{proof}

\begin{proposition}\label{prop:zero-NLCS}
$\{\dHzero\}$ is a family of NLCS Hamiltonians. 
\end{proposition}
\begin{proof}
    By definition, $\psi = \sum_j p_j\ketbra{\varphi_j}$, where each $\ket{\varphi_j}$ is a pure stabilizer state on $n$ qubits. The lower bound follows by applying Lemma \ref{lem:lb-zero-n-pure} to each $\ket{\varphi_j}$. Thus, every $n$-qubit stabilizer state has energy at least $\sin^2(\frac{\pi}{8})$ with respect to $\dHzero$, which implies $\dHzero$ is $\epsilon$-NLCS with $\epsilon = \sin^2(\frac{\pi}{8})$.
\end{proof}

\subsection{Towards NLACS}\label{app:NLACS-zero}

There are several notions of how ``non-Clifford'' a state is, the number of $\T$ gates being a common one. The notion we consider here is the number of arbitrary Pauli-rotation gates, $e^{i\theta P}$ for $\theta\in[0,2\pi)$ and $P\in\mcP_n$, as it encapsulates the $\T$ gate count.\footnote{The $\T$ gate is equal to $\T=\cos\left(\frac{\pi}{8}\right)\eye + i\sin\left(\frac{\pi}{8}\right)Z = e^{i\frac{\pi}{8} Z}$.}

\begin{lemma}\label{lem:rotation-gate-states}
Let $C$ be a quantum circuit on $n$-qubits containing polynomially many Clifford gates and at most $t$ arbitrary Pauli-rotation gates, $e^{i\theta_j P'_j}$. There exist $t$ Pauli operators, $\{P_j\}\subset\mcP_n$ and a stabilizer state $\ket\varphi$ such that
\begin{equation}\label{eq:ACS}
    C\ket0\n = \prod_{j\in[t]}\Big[e^{i\theta_j P_j}\Big]\ket\varphi,
\end{equation}
where by convention $C\ket0\n = \ket\varphi$ if $t=0$.
\end{lemma}
\begin{proof}
By definition we can decompose $C$ as
\begin{equation}
    C = C_{t}e^{i\theta_{t}P'_{t}}C_{t-1}\dots e^{i\theta_{2}P'_{2}}C_1 e^{i\theta_{1}P'_{1}} C_0,
\end{equation}
where each $C_\ell$ is a Clifford circuit. 

For every $j\in[t]$ we have $e^{i\theta_j P'_j} = \cos(\theta_j)\eye+i\sin(\theta_j)P'_j$.
Since Clifford gates normalize the Pauli group, for every Clifford circuit, $C'$, and every Pauli operator, $P'\in\mcP_n$, there is another Pauli operator, $P''\in\mcP_n$, such that $C'(\cos\theta\eye+i\sin\theta P') = (\cos\theta\eye+i\sin\theta P'')C' $. Thus, we can move each Clifford circuit, $C_\ell$, past all of the Pauli-rotation gates by changing only the individual Pauli operators via the conjugation relations of $C_\ell$.

Ultimately, we can rewrite $C$ as
\begin{equation}
    C = e^{i\theta_{t}P_{t}}\dots e^{i\theta_{2}P_{2}}e^{i\theta_{1}P_{1}}C_{t}\dots C_1 C_0,
\end{equation}
for $t$ Pauli operators, $\{P_t\}$, as desired.

\end{proof}

Proposition \ref{prop:zero-NLCS} shows that the $D$-rotated zero Hamiltonian, $\Tilde{\ham}_0 = \frac{1}{n}\sum\big(D^\dagger\ketbra{1}D\big)_i$, is $\sin^2\left(\frac{\pi}{8}\right)$-NLCS. It is natural to question if $\Tilde{\ham}_0$ is also $\epsilon$-NLACS for some appropriate constant $\epsilon$. In this section we will prove an explicit lower-bound on all states prepared by Clifford gates + at most 1 Pauli-rotation gate:
\begin{equation}
    \bra{\psi}\dHzero\ket{\psi}\geq \left(1-\frac{1}{n}\right)\sin^2\left(\frac{\pi}{8}\right).
\end{equation}
In fact, there is numerical evidence suggesting the following lower bound for an arbitrary number of Pauli-rotation gates, though we have been unable to prove it analytically:
\begin{conjecture}\label{conj:lb-t-T-gates}
Let $\ket\psi$ be an $n$-qubit state prepared by a Clifford circuit plus at most $t$ Pauli-rotation gates. For the $D$-rotated zero-Hamiltonian, $\dHzero$, the energy of $\ket\psi$ is lower-bounded as
\begin{equation}
    \bra{\psi}\dHzero\ket{\psi}\geq \left(1-\frac{t}{n}\right)\sin^2\left(\frac{\pi}{8}\right).
\end{equation}
In particular, if there is a constant $\beta\in [0,1)$ such that $t\leq \beta n$ for all sufficiently large $n$, then the energy of $\ket\psi$ is lower-bounded by $\left(1-\beta \right)\sin^2\left(\frac{\pi}{8}\right)>0$, a constant.
\end{conjecture}
By Lemma \ref{lem:rotation-gate-states}, the most general such state is a stabilizer state with $t$ Pauli-rotation gates applied to it and no intermediate circuits between them. The intuition behind Conjecture \ref{conj:lb-t-T-gates} is that the only way to reduce the energy of a stabilizer state is to ``undo'' one of the $D$ gates conjugating the Hamiltonian. For instance, to produce a state with sub-constant energy one could apply $n-o(n)$ $D$ gates to $\ket0\n$.

We note also that is in unclear what, if any, similar lower bound could be shown for an arbitrary $D$-rotated CSS Hamiltonian (as considered in Theorem \ref{thm:CSS-NLCS}). We leave this as an open problem, as well. For now, we consider the case of $t=1$ for the $D$-rotated zero Hamiltonian.

First, recall the following definition.
\stabsubgroups*

The following lemma gives an explicit description of the local density matrices of states with at most 1 Pauli-rotation gate.
\begin{lemma}\label{lem:single-R-reduced-state}
Let $\ket\psi=e^{i\theta P}\ket\varphi$ for $P\in\mcP_n$, $\theta\in[0,2\pi)$, and let $\ket\varphi$ be a stabilizer state with $G\equiv\stab(\ket\varphi)$. For $A\subset[n]$ we can write $\psi_A \equiv \Tr_{-A}[\ketbra{\psi}]$ as 
\begin{equation}
    \psi_A = \frac{1}{2^\abs{A}}\sum_{\hat g\in G_A} \bigg(\cos^2(\theta)\hat g +\sin^2(\theta)P_A \hat g P_A\bigg)+\frac{1}{2^\abs{A}}\sum_{g'\in G_{A,P}} i\sin(\theta)\cos(\theta) [P_A, g'].
\end{equation}
 The left part of this expression can be thought of as the stabilizer part of $\psi_A$, as it is the convex combination of two stabilizer states, and the right hand part can be thought of as the non-stabilizer part, as it equals zero if $P\in G$ or if $P_A=\eye$.
\end{lemma}
\begin{proof}
Since $\ket\varphi$ is a stabilizer state there is a stabilizer group $G$ with $\abs{G}=2^n$ such that $\ketbra{\varphi}=\frac{1}{2^n}\sum_{g\in G} g$. Using the exponential of Pauli matrices we have
\begin{align}
    \psi &= \frac{1}{2^n}\sum_{g\in G}(\cos(\theta)\eye +i\sin(\theta)P)g(\cos(\theta)\eye -i\sin(\theta)P), \\
    &= \frac{1}{2^n}\sum_{g\in G} \cos^2(\theta)g +\sin^2(\theta)PgP +i\sin(\theta)\cos(\theta)Pg -i\sin(\theta)\cos(\theta)gP, \\
    &= \frac{1}{2^n}\sum_{g\in G} \bigg(\cos^2(\theta)g +\sin^2(\theta)PgP\bigg) + \frac{1}{2^n}\sum_{g\in G} \bigg(i\sin(\theta)\cos(\theta) (Pg-gP)\bigg). \label{eq:psi-single-R}
\end{align}
Consider tracing out all qubits outside of the set $A$. The only Pauli group element with nonzero trace is $\eye$, which has trace 2. For the left term in Equation \eqref{eq:psi-single-R}, we have
\begin{align}
    &\frac{1}{2^n}\sum_{g\in G}\bigg( \cos^2(\theta)\Tr_{-A}[g] +\sin^2(\theta)\Tr_{-A}[PgP] \bigg)\\
    &= \frac{1}{2^n}\sum_{g\in G}\bigg( \cos^2(\theta)g_A\prod_{j\in[n]\setminus A}\Tr[g_j] +\sin^2(\theta)P_A g_A P_A \prod_{j\in[n]\setminus A}\Tr[P_j g_j P_j]\bigg), \\
    &= \frac{1}{2^n}\sum_{g\in G} \bigg(\cos^2(\theta)g_A +\sin^2(\theta)P_A g_A P_A\bigg) \bigg(\prod_{j\in[n]\setminus A}\Tr[g_j]\bigg), \\
    &= \frac{1}{2^\abs{A}}\sum_{\hat g\in G_A} \bigg(\cos^2(\theta)\hat g +\sin^2(\theta)P_A \hat g P_A\bigg),
\end{align}
where the last line follows since only those $g\in G$ which are identity outside of $A$ will have nonzero trace, and the product of the individual traces when non-zero is $2^{n-\abs{A}}$.

Similarly, for the right term in Equation \eqref{eq:psi-single-R} we have
\begin{align}
    &\frac{1}{2^n}\sum_{g\in G} \bigg(i\sin(\theta)\cos(\theta) \Tr_{-A}[Pg-gP]\bigg), \\
    &= \frac{1}{2^n}\sum_{g\in G} \bigg(i\sin(\theta)\cos(\theta) [P_A, g_A] \bigg)\bigg(\prod_{j\in[n]\setminus A}\Tr[P_j g_j]\bigg), \\
    &= \frac{1}{2^\abs{A}}\sum_{g'\in G_{A,P}} i\sin(\theta)\cos(\theta) [P_A, g'],
\end{align}
where the last line follows again since the trace will be non-zero only if $g_j=P_j$ for all $j\notin A$. 
\end{proof}

\begin{lemma}\label{lem:lb-1-T-gate}
\begin{equation}
    \bra{\psi}\dHzero\ket{\psi}\geq \left(1-\frac{1}{n}\right)\sin^2\left(\frac{\pi}{8}\right).
\end{equation}
\end{lemma}

\begin{proof}
    By Lemma \ref{lem:rotation-gate-states} there is a Pauli operator, $P$, and an $n$-qubit Clifford state $\ket\varphi$ such that $\ket\psi = e^{i\theta P}\ket\varphi$. Let $G\equiv\stab(\ket\varphi)$.

    Recall that by definition $\dHzero =  \frac{1}{n}\sum_{i=1}^n {{\Tilde{\ham}_0}^{(1)}}\mid_i\otimes \eye_{[n]\setminus\{i\}}$, so
\begin{align}
    \bra{\psi}\dHzero\ket{\psi} &=  \frac{1}{n}\sum_{i=1}^n \Tr[{\psi_i{\Tilde{\ham}_0}^{(1)}}],
\end{align}
where $\psi_i\equiv \Tr_{-i}[\ketbra{\psi}]$ is the reduced state of $\ket{\psi}$ on qubit $i$. We will show that at most one of the terms in this summation can be 0, and that the remainder of the terms are lower-bounded by $\sin^2\left(\frac{\pi}{8}\right)$.

By Lemma \ref{lem:single-R-reduced-state} we can write the reduced state as
\begin{equation}
    \psi_i = \frac{1}{2}\sum_{\hat g\in G_i} \bigg(\cos^2(\theta)\hat g +\sin^2(\theta)P_i \hat g P_i\bigg)+\frac{1}{2}\sum_{g'\in G_{i,P}} i\sin(\theta)\cos(\theta) [P_i, g'].
\end{equation}

We proceed in cases:
\begin{enumerate}[label=\textbf{\Roman*.}]
    \item If $P\in G$, $P_i=\eye$, $G_{i,P}=\emptyset$, or $G_{i,P}=\{\eye\}$ then $\psi_i$ is a stabilizer state, so $ \Tr[{\psi_i{\Tilde{\ham}_0}^{(1)}}]\geq\sin^2\left(\frac{\pi}{8}\right) $.
    \item Suppose the four conditions from Case \textbf{I.} do not hold. It must be that $G_{i,P}=\{\eye,P^\star\}$ for some $P^\star\in\mcP_1\setminus\{\eye,P_i\}$; $P^\star$ cannot be $P_i$ as this would imply $P\in G$. Note that $G_{i,P}$ cannot be any larger as this would contradict the fact $G$ is a stabilizer group. We now consider cases for $G_i$.
    \begin{enumerate}[label*=\textbf{\arabic*.}]
        \item If $G_i=\{\eye\}$, then $\psi_i$ can be written as
        \begin{align}\label{eq:single-R-reduced-state}
             \psi_i &= \frac{1}{2}\eye+\frac{1}{2} i\sin(\theta)\cos(\theta) [P_i, P^\star], \\
             & =  \frac{1}{2}\eye+\frac{1}{4}\sin(2\theta)\sigma,
        \end{align}
        since $P_i\neq P^\star$ and $2i[P_i,P^\star]=\sigma$ for some non-identity Pauli. The bound
        $ \Tr[{\psi_i{\Tilde{\ham}_0}^{(1)}}]\geq\sin^2\left(\frac{\pi}{8}\right) $ holds by direct computation over $\sigma\in\mcP\setminus\{\pm\eye\}$.
        \item If $G_i$ is non-trivial then $G_i=\{\eye,P^\star\}$ since it must commute with the $g\in G$ which satisfies $g_i=P^\star$ and $g_{-i}=P_{-i}$ (which exists since we are in Case \textbf{II.}) Since $P^\star\notin\{\eye,P_i\}$ we can write $\psi_i$ as
        \begin{align}
            \psi_i &= \frac{1}{2}\eye + \frac{1}{2} \big(\cos^2(\theta) -\sin^2(\theta)\big)P^\star +\frac{1}{2} i\sin(\theta)\cos(\theta) [P_i, P^\star], \\
            & = \frac{1}{2}\eye + \frac{1}{2} \cos(2\theta) P^\star +\frac{1}{4} \sin(2\theta)i[P_i, P^\star].
        \end{align}
        By direct computation we have the following:
        \begin{enumerate}[label*=\textbf{\alph*.}]
            \item If $P_i \neq Y$ then $ \Tr[{\psi_i{\Tilde{\ham}_0}^{(1)}}]\geq\sin^2\left(\frac{\pi}{8}\right) $ regardless of $\theta$.
            \item If $P_i = Y$ and $P^\star\neq Z$ then $ \Tr[{\psi_i{\Tilde{\ham}_0}^{(1)}}]\geq\sin^2\left(\frac{\pi}{8}\right) $ regardless of $\theta$.
            \item If $P_i = Y$ and $P^\star= Z$ then $\Tr[{\psi_i{\Tilde{\ham}_0}^{(1)}}]\geq 0$ with possible equality.
        \end{enumerate}

    \end{enumerate}
\end{enumerate}
To recap the cases, $\psi_i$ can have energy less than $\sin^2\left(\frac{\pi}{8}\right)$ only if
(1) $P_i = Y$, (2) $Z_i\in G$, and (3) there is a $g\in G$ such that $g_i=Z$ and $g_{-i}=P_{-i}$, i.e. $g$ and $P$ agree on every qubit except $i$.

We must show that at most one qubit can satisfy all three of these condition for a given $P\in\mcP_n$ and stabilizer group $G$. Suppose there are two such qubits, $i$ and $j$, which satisfy (1) $P_i=P_j=Y$, (2) $Z_i,Z_j\in G$, and (3) there exist $g,h\in G$ such that $g_i=h_j=Z$, $g_{-i}=P_{-i}$, and $h_{-j}=P_{-j}$. By condition (3) $g_i=Z$ and $g_j=Y$ and by condition (2) $Z_j\in G$, but this implies that $gZ_j = -Z_j g$, which contradicts the fact that $G$ is abelian. Thus, at most a single qubit can satisfy the conditions required for the reduced state $\psi_i$ to have energy less than $\sin^2\left(\frac{\pi}{8}\right)$, which implies the desired lower bound.

\end{proof}

\end{document}